\documentclass[11pt,reqno,a4paper]{amsart}
\usepackage[utf8]{inputenc}
\usepackage{mathtools,amsmath,amssymb}
\usepackage[T1]{fontenc}
\usepackage{lmodern}
\usepackage{microtype}
\usepackage{hyperref}
\usepackage{xspace}
\usepackage{bibentry}
\usepackage{easybmat}
\usepackage{mathdots}

\usepackage{cite}

\usepackage{subcaption}

\usepackage{bibentry}
\usepackage{tikz}

\usepackage{dynkin-diagrams}

\usepackage[margin=2.4cm]{geometry}

\usepackage{tikz}
\usetikzlibrary{decorations.pathreplacing}

\numberwithin{equation}{section}

\usepackage{amsmath}
\usepackage{amsthm}

\makeatletter
\let\@wraptoccontribs\wraptoccontribs
\makeatother

\theoremstyle{plain}
\newtheorem{Thm}[equation]{Theorem}

\newtheorem{Lem}[equation]{Lemma}
\newtheorem{Prop}[equation]{Proposition}
\newtheorem{Conj}[equation]{Conjecture}

\newtheorem*{Thm-non}{Theorem}

\theoremstyle{definition}

\newtheorem{Rem}[equation]{Remark}

\newcommand{\osc}[1]{\mathbf{#1}}

\newcommand{\oa}{\osc{a}}
\newcommand{\ob}{\osc{b}}
\newcommand{\oc}{\osc{c}}

\newcommand{\oad}{\osc{\bar{a}}}
\newcommand{\obd}{\osc{\bar{b}}}

\newcommand{\ocd}{\osc{\bar{c}}}

\newcommand{\vp}{\bar{\mathbf{v}}}
\newcommand{\vm}{\mathbf{v}}

\renewcommand{\wp}{\bar{\mathbf{w}}}
\newcommand{\wm}{\mathbf{w}}

\newcommand{\up}{\bar{\mathbf{u}}}
\newcommand{\um}{\mathbf{u}}

\newcommand{\ap}{\bar{\mathbf{A}}}
\newcommand{\am}{{\mathbf{A}}}
\newcommand{\bp}{\bar{\mathbf{B}}}
\newcommand{\bm}{{\mathbf{B}}}

\newcommand{\kp}{\bar{\mathbf{K}}}
\newcommand{\km}{{\mathbf{K}}}
\newcommand{\qp}{\bar{\mathbf{Q}}}
\newcommand{\qm}{{\mathbf{Q}}}
\newcommand{\kpp}{\kp^\circ}
\newcommand{\kmm}{\km^\circ}

\newcommand{\idb}{\JD }
\newcommand{\id}{{\ID }}
\newcommand{\idg}{{\mathrm{G}} }

\newcommand\iso{\,\vphantom{j^{X^2}}\smash{\overset{\sim}{\vphantom{\rule{0pt}{0.20em}}\smash{\longrightarrow}}}\,}

\DeclareMathOperator{\diag}{diag}

\newmuskip\pFqmuskip
\newcommand*\pFq[6][8]{%
    \begingroup
    \pFqmuskip=#1mu\relax
    \mathcode`\,=\string"8000
    \begingroup\lccode`\~=`\,
    \lowercase{\endgroup\let~}\pFqcomma
    {}_{#2}F_{#3}{\left(\genfrac..{0pt}{}{#4}{#5};#6\right)}%
    \endgroup}
\newcommand{\pFqcomma}{\mskip\pFqmuskip}

\def\be{\begin{eqnarray}}
\def\ee{\end{eqnarray}}

\newcommand{\gr}[1]{{ | {#1} | }}
\newcommand{\sfrac}[2]{{\textstyle\frac{#1}{#2}}}

\newcommand{\Qop}{\mathrm{Q}}
\newcommand{\Pop}{\mathrm{P}}
\newcommand{\ID}{\mathrm{Id}}
\newcommand{\JD}{\mathrm{J}}

\newcommand{\rtt}{\mathrm{rtt}}
\newcommand{\End}{\mathrm{End}}

\newcommand{\gl}{\mathfrak{gl}}
\newcommand{\ssl}{\mathfrak{sl}}

\newcommand{\BZ}{\mathbb{Z}}
\newcommand{\BC}{\mathbb{C}}

\newcommand{\bL}{\bar{L}}

\makeatletter
\g@addto@macro\bfseries{\boldmath}
\makeatother

\makeatletter
\newcases{rrcases}{\quad}{%
  \hfil$\m@th\displaystyle{##}$}{$\m@th\displaystyle{##}$\hfil}{\lbrace}{.}
\makeatother

\newcommand{\fosp}{\mathfrak{osp}}

\newcommand{\VV}{{\mathsf{V}}}
\newcommand{\sfv}{{\mathsf{v}}}

\newcommand{\Parity}{\Upsilon_V}
\newcommand{\parity}{\Upsilon_{\mathsf{V}}}

\newcommand{\ossc}{\mathrm{osc}}

\begin{document}

\title[Orthosymplectic superoscillator Lax matrices]
{\large{\textbf{Orthosymplectic superoscillator Lax matrices}}}

\author{Rouven Frassek}
\address{R.F.: University of Modena and Reggio Emilia, FIM, Via G.~Campi~213/b, 41125 Modena, Italy}
\email{rouven.frassek@unimore.it}

\author{Alexander Tsymbaliuk}
\address{A.T.: Purdue University, Department of Mathematics, West Lafayette, IN 47907, USA}
\email{sashikts@gmail.com}

\begin{abstract}
We construct Lax matrices of superoscillator type that are solutions of the RTT-relation for the rational
orthosymplectic $R$-matrix, generalizing orthogonal and symplectic oscillator type Lax matrices previously constructed
by the authors in~\cite{f,ft1,fkt}. We further establish factorisation formulas among the presented solutions.
\end{abstract}

\maketitle
\tableofcontents


\section{Introduction}\label{sec:intro}


\subsection{Summary}\label{ssec:summary}
\

The study of supersymmetric solutions to the Yang-Baxter equation goes back to the works of Kulish and Sklyanin in the
early 80's, see e.g.~\cite{KS}, who introduced the $R$-matrix that generates the supersymmetric Yangian of $\gl(n|m)$,
see~\cite{N}. As common for Lie superalgebras, the underlying vector space is equipped with a $\BZ_2$-grading to
incorporate bosonic and fermionic degrees of freedom. Similarly to the purely bosonic case, there exists an evaluation map
  $\mathrm{ev}\colon Y(\gl(n|m))\to U(\gl(n|m))$
from the Yangian of $\gl(n|m)$ into the universal enveloping algebra of the Lie superalgebra $\gl(n|m)$, see~\cite{K,N},
which facilitates the study of the spectrum of supersymmetric spin chains. In particular, the algebraic Bethe ansatz for
a large class of representations in the quantum space and the construction of the corresponding transfer matrices are
well-understood, see~\cite{K}. The same is true for the functional relations ($T$-systems and $QQ$-systems) among such
transfer matrices and $Q$-operators, see \cite{tsuboi1,tsuboi2,KSZ,tsuboi3} as well as \cite{kuniba} for an overview.

The construction of the $Q$-operators has been carried out more recently in \cite{Frasseksusy} employing certain degenerate
Lax matrices of superoscillator type (see also \cite{klt} for a different approach). The Lax matrices for $Q$-operators in
the trigonometric case were obtained in \cite{Bazhanov:2008yc} for $n+m=3$ and in \cite{Tsuboi:2012sz,Tsuboi:2019vvv}
for arbitrary $n,m$.
The degenerate solutions of the Yang-Baxter equation with a rational (resp.\ trigonometric) $R$-matrix arise naturally by
taking certain limits of the evaluation representation of the Yangian (resp.\ quantum affine algebra) in the parabolic Verma
modules of the underlying Lie algebra, realized in terms of superoscillator algebras.
The crucial difference between the solutions that arise from the ordinary Yangian and the degenerate solutions is that
the coefficients of the leading power of the spectral parameter are not of full rank in degenerate case. Therefore,
this class of solutions does not arise through the ordinary Yangian, but is rather related to the so-called (RTT antidominantly)
shifted Yangians, as has been recently realized in~\cite{fpt} (cf.~\cite{cgy} for an interpretation via the $4d$ Chern-Simons theory).
The latter are usually defined in terms of Drinfeld's current realization, see e.g.~\cite[Appendix B]{BFN}, and
the identification with the aforementioned RTT ones goes through the Gauss decomposition of the generating matrix $T(x)$ as
in the ordinary case~\cite{bk}, see~\cite[Theorem 2.54]{fpt}. For the bosonic case of $\mathfrak{gl}(n)$, the degenerate Lax matrices
that are linear in the spectral parameter were constructed in \cite{fp}, while the reconstruction of degenerate Lax matrices
at any order of the spectral parameter  has been achieved in \cite{fpt} using the results of \cite{BFN}. Using
the $S(n)$-invariance of the rational $R$-matrix of $\gl_n$ one can further obtain other degenerate Lax matrices which
do not admit a Gauss decomposition (and thus are no longer directly related to the \emph{shifted Yangians}).
Thus, the transfer matrices are constructed from representations of the ordinary Yangian,
while the $Q$-operators are constructed from representations of the shifted Yangian, cf.~\cite{hz} and references therein.
We note that this approach allows to deduce functional relations among transfer matrices and $Q$-operators directly from
the Yang-Baxter equation using certain factorisation properties of the Lax matrices combined with the BGG-type resolutions
for the underlying Lie algebras, see \cite{BLZ,BHK,fkt}.

The situation changes drastically for the Yangians of the orthosymplectic Lie superalgebras $\fosp(N|2m)$ which unify
the bosonic cases in $BCD$-types, that is, Yangians of orthogonal $\mathfrak{so}_N$ and symplectic $\mathfrak{sp}_{2m}$
Lie algebras. Similarly to $BCD$-type, the evaluation map no longer exists and representations of the orthosymplectic
Lie superalgebra cannot be always lifted to representations of the corresponding Yangian. The $R$-matrix in the defining
vector representation was obtained in \cite{Arnaudon}, generalizing the $BCD$-type $R$-matrix of~\cite{zz}, but other
solutions of the Yang-Baxter relation or the RTT-relation (for the Lax matrix) are scarce, see \cite{fikk,ikk} and
references therein. The algebraic Bethe ansatz has been obtained for spin chains in the defining vector representation
in \cite{GM}, but little is known about other representations as well as the underlying  functional relations.
A glimpse towards the latter appeared in the study~\cite{bcf} of the AdS/CFT-correspondence in relation to the quantum
spectral curve for $AdS_4/CFT_3$, but remain to be confirmed from the first principles.
Further results for more general Lie superalgebras were recently obtained in \cite{Tsuboi:2023}.
The full understanding of the $QQ$-system may yield new methods of solving Bethe equations similar to the method
developed in \cite{MV} for $\gl(n|m)$.

In this paper, we enlarge the class of representations of the orthosymplectic (shifted) Yangians by introducing several
Lax matrices of superoscillator type, which can be used to construct transfer matrices and $Q$-operators. We anticipate
the BGG-type functional relations among those, generalizing our recent $BCD$-type results of~\cite{fkt}.

Finally, let us note that there has been an increased mathematical interest in the theory of quantum supergroups.
In the context of orthosymplectic Yangians of $\fosp(N|2m)$ specifically, their Drinfeld realizations were established
recently in~\cite{mr} and~\cite{m} for the cases $N=1$ and $N>3$ with the standard parity, respectively.
In the sequel paper~\cite{ft2}, we present uniform Drinfeld realizations of orthosymplectic Yangians for any $N,m$
and, most importantly, any parity sequence.


\subsection{Outline}\label{ssec:outline}
\

The structure of the present paper is the following:

\medskip
\noindent
$\bullet$
In Section~\ref{sec:Atype}, we recall the key results of~\cite{Frasseksusy} on $\gl(n|m)$-type Lax matrices
that serve as motivation and prototype for our new constructions in the orthosymplectic type.

\medskip
\noindent
$\bullet$
In Section~\ref{sec:setup-and-invariance}, we set the notation for the orthosymplectic $R$-matrix and the corresponding
Lax matrices, as well as discuss the invariance of this $R$-matrix that is needed for our latter results.

\medskip
\noindent
$\bullet$
In Section~\ref{sec:linear-Lax}, we construct some linear superoscillator Lax matrices of orthosymplectic type.
First, we construct a degenerate linear Lax matrix in Theorem~\ref{thm:degenerate-linear-Lax}. Fusing two of those,
we then construct a non-degenerate linear Lax matrix in Proposition~\ref{prop:lin}, whose normalized limits recover
back the degenerate Lax matrices, see Remark~\ref{rem:normalized-limit-osp}. In the special cases $m=0$ or $n=0$,
we recover the corresponding orthogonal and symplectic Lax matrices of~\cite{f,ft1,fkt}, respectively.

\medskip
\noindent
$\bullet$
In Section~\ref{sec:quadratic-Lax}, we investigate some quadratic orthosymplectic Lax matrices of superoscillator type.
First, fusing two degenerate linear Lax matrices from Section~\ref{sec:linear-Lax}, we construct a Lax matrix of size
$(N+2m)\times(N+2m)$ for even $N$ in Theorem~\ref{prop:quadL}. We call this matrix degenerate quadratic Lax as its diagonal
is $\sim(x^2,x,\ldots,x,1)$ with respect to the spectral parameter $x$. A further degeneration of this matrix, depending
only in $N+2m-2$ pairs of superoscillators, is obtained in Proposition~\ref{prop:quadratic-degenerate-specialized}.
A similar formula provides an orthosymplectic Lax matrix for odd $N$, see Conjecture~\ref{conj:odd-N}. Finally, fusing
two degenerate quadratic Lax matrices we derive explicit non-degenerate quadratic orthosymplectic Lax matrices in
Proposition~\ref{prop:nondeg-quadratic-small} and Theorem~\ref{prop:nondeg-quadratic-big},
whose normalized limits recover back the degenerate quadratic Lax matrices, see Remark~\ref{rem:normalized-limit-osp-quadratic}.

\medskip
\noindent
$\bullet$
In Appendix~\ref{sec:twists}, we present explicit formulas~(\ref{eq:T-twist-even},~\ref{eq:T-twist-odd})
and~(\ref{eq:Q-twist-1},~\ref{eq:Q-twist-2},~\ref{eq:Q-twist-3}) for the twists needed to define both
the transfer matrices and the $Q$-operators, as mentioned in the Introduction above.

\medskip
\noindent
While the constructions of Sections~\ref{sec:linear-Lax}--\ref{sec:quadratic-Lax} are presented for the specific parity
sequences~\eqref{eq:our-osp-grading} or~\eqref{eq:our-osp-grading-odd}, similar Lax matrices exist for other parity sequences as well,
according to Remarks~\ref{rem:different parities osp},~\ref{rem:linear-Lax-all-parities},~\ref{rem:quadratic-Lax-all-parities}.


\subsection{Acknowledgement}\label{ssec:acknowl}
\

R.F.\ thanks the organizers of the \textit{10th Bologna Workshop on Conformal Field Theory and Integrable Models}
where the results of this note were presented (the slides are available at~\cite{ftalk}).

The work of R.F.\ was partially supported in part by the INFN grant Gauge and String Theory (GAST),
by the PRIN project CUP-E53D23002220006 (Next Generation EU), the FAR UNIMORE project CUP-E93C23002040005,
and by the INdAM–GNFM project codice CUP-E53C22001930001.
The work of A.T.\ was partially supported by NSF Grants DMS-$2037602$ and DMS-$2302661$.

We are grateful to the anonymous referees for useful suggestions that improved the exposition.


\section{General linear Lax matrices}\label{sec:Atype}

The first results on superoscillator type Lax matrices for $Q$-operators and their factorisation formulas were presented
in \cite{Frasseksusy} for the rational $R$-matrices of $\gl(n|m)$-type. The Lax matrices for the trigonometric case were
obtained in \cite{Bazhanov:2008yc} for $n+m=3$ and in \cite{Tsuboi:2012sz,Tsuboi:2019vvv} for arbitrary $n,m$.
In this section, we briefly recall the results of \cite{Frasseksusy} that are relevant to the rest of the paper.

Fix $n,m\geq 0$ and consider a superspace $\VV=\VV_{\bar{0}}\oplus \VV_{\bar{1}}$ with a $\BC$-basis
$\sfv_1,\ldots,\sfv_{n+m}$ such that each $\sfv_i$ is either \emph{even} (that is, $\sfv_i\in \VV_{\bar{0}}$) or
\emph{odd} (that is, $\sfv_i\in \VV_{\bar{1}}$) and $\dim(\VV_{\bar{0}})=n, \dim(\VV_{\bar{1}})=m$. For
$1\leq i\leq n+m$, we define $|i|:=|\sfv_i|\in \BZ_2$. We define the \emph{parity sequence} associated to $\VV$ via
\begin{equation}\label{eq:parity-sequence}
  \parity:=\big(|\sfv_1|,\ldots,|\sfv_{n+m}|\big)\in \big\{\bar{0},\bar{1}\big\}^{n+m} \,.
\end{equation}
For a superalgebra $A$ and homogeneous elements $a,a'\in A$, their \emph{supercommutator} is defined as
\begin{equation}\label{eq:super commutator}
  [a,a']=aa'-(-1)^{|a| |a'|}\, a'a \,,
\end{equation}
where $|a|$ denotes the $\BZ_2$-grading of $a$. Given two superspaces $A=A_{\bar{0}}\oplus A_{\bar{1}}$
and $B=B_{\bar{0}}\oplus B_{\bar{1}}$, their tensor product $A\otimes B$ is also a superspace with
  $(A\otimes B)_{\bar{0}}=A_{\bar{0}}\otimes B_{\bar{0}}\oplus A_{\bar{1}}\otimes B_{\bar{1}}$
and
  $(A\otimes B)_{\bar{1}}=A_{\bar{0}}\otimes B_{\bar{1}}\oplus A_{\bar{1}}\otimes B_{\bar{0}}$.
Furthermore, if $A$ and $B$ are superalgebras, then $A\otimes B$ is also a superalgebra,
called the \emph{graded tensor product} of $A$ and $B$, with the multiplication defined by
\begin{equation}\label{eq:graded tensor product}
  (a\otimes b)(a'\otimes b')=(-1)^{|b| |a'|}\, (aa')\otimes (bb')
\end{equation}
for any homogeneous elements $a\in A_{|a|}, a'\in A_{|a'|}, b\in B_{|b|}, b'\in B_{|b'|}$.

Let $\Pop\colon \VV\otimes \VV\to \VV\otimes \VV$ be the permutation operator defined by
\begin{equation}\label{eq:P}
  \Pop \, = \sum_{i,j=1}^{n+m} (-1)^{\gr{j}}\, e_{ij}\otimes e_{ji} \,,
\end{equation}
whose action is explicitly given by:
\begin{equation}\label{eq:Pop-explicit}
  \Pop(\sfv_j\otimes \sfv_i)=(-1)^{\gr{i} \gr{j}}\, \sfv_i\otimes \sfv_j \,.
\end{equation}
Consider the corresponding \emph{rational $R$-matrix} (of general linear type, or super $A$-type for short):
\begin{equation}\label{eq:Atype-Rmatrix}
  R(x)=\mathsf{R}^\VV(x)=x\ID+\Pop \,,
\end{equation}
which satisfies the famous \emph{Yang-Baxter equation} (with a spectral parameter):
\begin{equation}\label{eq:YBE}
  R_{12}(x)R_{13}(x+y)R_{23}(y)=R_{23}(y)R_{13}(x+y)R_{12}(x) \,.
\end{equation}
For any superalgebra $A$, an even matrix $L(x)=L^\VV(x)=(L_{ij}(x))_{i,j=1}^{n+m}\in \End\,\VV\otimes A[[x,x^{-1}]]$
is called an ($A$-valued) \emph{Lax matrix} if it satisfies the corresponding \emph{RTT-relation} with $R(x)$
of~\eqref{eq:Atype-Rmatrix}
\begin{equation}\label{eq:rttA}
  R_{12}(x-y)L_1(x)L_2(y)=L_2(y)L_1(x)R_{12}(x-y) \,,
\end{equation}
viewed as an equality in $\End\,\VV \otimes \End\,\VV \otimes A[[x,y,x^{-1},y^{-1}]]$. Coefficient-wise,
the equation~\eqref{eq:rttA} is equivalent to (see~\cite{gowthesis}) the well-known system of equations
for all $1\leq i,j,k,\ell \leq n+m$:
\begin{equation}\label{eq:yangianalgA}
  [L_{ij}(x),L_{k\ell}(y)]=
  \frac{(-1)^{\gr{i}\gr{j}+\gr{i}\gr{k}+\gr{j}\gr{k}}}{x-y} \Big(L_{kj}(y)L_{i\ell}(x)-L_{kj}(x)L_{i\ell}(y)\Big) \,.
\end{equation}

\begin{Rem}\label{rem:sign convention}
(a) Here, we identify the matrix $(L_{ij}(x))_{i,j=1}^{n+m}$ with
  $\sum_{i,j=1}^{n+m} (-1)^{|i| |j|+|j|}\, e_{ij}\otimes L_{ij}(x)$.
Evoking~\eqref{eq:graded tensor product}, the extra sign $(-1)^{|i| |j|+|j|}$ ensures that the product
of matrices is calculated as usual. The above ``even'' condition means that $\BZ_2$-grading of all
coefficients of $L_{ij}(x)$ is $|i|+|j|$.

\medskip
\noindent
(b) The data of a Lax matrix $L^\VV(x)=(L_{ij}(x))_{i,j=1}^{n+m}$ with $L_{ij}(x)\in \delta_{ij}+x^{-1}A[[x^{-1}]]$
for all $i,j$ is equivalent to an algebra homomorphism $Y^\rtt(\gl(\VV))\to A$ from an RTT super Yangian of $\gl(\VV)$.

\medskip
\noindent
(c) Unless we want to emphasize the dependence on $\parity$, the superscript $\VV$ will be ignored.
\end{Rem}

\begin{Rem}\label{rem:different parities}
Let $\tilde{\VV}$ be another superspace with a $\BC$-basis $\tilde{\sfv}_1,\ldots,\tilde{\sfv}_{n+m}$
such that each $\tilde{\sfv}_i$ is even or odd and
  $\dim(\VV_{\bar{0}})=\dim(\tilde{\VV}_{\bar{0}}), \dim(\VV_{\bar{1}})=\dim(\tilde{\VV}_{\bar{1}})$.
Pick a permutation $\sigma\in S(n+m)$ such that $\sfv_i\in \VV$ and $\tilde{\sfv}_{\sigma(i)}\in \tilde{\VV}$
have the same $\BZ_2$-grading for all $1\leq i\leq n+m$, and define a superspace isomorphism
$\idb_\sigma\colon \VV \iso \tilde{\VV}$ via $\sfv_i\mapsto \tilde{\sfv}_{\sigma(i)}$.
The corresponding $R$-matrices~\eqref{eq:Atype-Rmatrix} are related via
\begin{equation}\label{eq:Rmatrix-intertwiner}
  \mathsf{R}^{\tilde{\VV}}(x)=(\idb_\sigma\otimes \idb_\sigma)\, \mathsf{R}^\VV(x)\, (\idb_\sigma\otimes \idb_\sigma)^{-1} \,.
\end{equation}
As a result, if $L^\VV(x)=(L_{ij}(x))_{i,j=1}^{n+m}$ is a solution of~\eqref{eq:rttA}, then
$L^{\tilde{\VV}}(x):=\idb_\sigma L^\VV(x) \idb_\sigma^{-1}$ is a solution of~\eqref{eq:rttA}
for $\tilde{\VV}$ used instead of $\VV$. In other words, having constructed some Lax matrices, a natural
$S(n+m)$-symmetry allows for analogous Lax matrices for all parity sequences~\eqref{eq:parity-sequence}.
\end{Rem}

Among other Lax matrices, the following family $\{L^\VV_a(x)\}_{a=0}^{n+m}$ was constructed in
\cite{Frasseksusy}\footnote{It is obtained from the \emph{linear canonical} $\mathbf{L}$-operator of
\cite[(2.20)]{Frasseksusy} for the trivial representation of the additional generators of $\gl(p|q)$ in
\emph{loc.cit.}, i.e.\ setting $E_{ab}\mapsto 0$, and an appropriate shift of the spectral parameter $x$.\label{fn1}}:
\begin{equation}\label{eq:one matrices1}
  L_a(x) = L^\VV_a(x) = \,
  \left(\begin{BMAT}[5pt]{c:c}{c:c}
    x\ID_{a}-\kp\km & \kp \\
    -\km & \ID_{n+m-a}
   \end{BMAT}\right)
\end{equation}
with
\begin{equation}\label{eq:kpA}
  \kp =
  \left(\begin{array}{cccc}
    \bar\xi_{1,a+1} & \bar\xi_{1,a+2} & \cdots & \bar \xi_{1,n+m} \\
    \bar\xi_{2,a+1} & \bar\xi_{2,a+2} & \cdots & \bar \xi_{2,n+m} \\
    \vdots & \vdots & \ddots& \vdots \\
    \bar\xi_{a,a+1} & \bar\xi_{a,a+2} & \cdots & \bar \xi_{a,n+m}
  \end{array}\right)
\end{equation}
and
\begin{equation}\label{eq:kmA}
  \km =
  \left(\begin{array}{cccc}
     (-1)^\gr{1} \xi_{a+1,1} & (-1)^\gr{2} \xi_{a+1,2} & \cdots & (-1)^\gr{a} \xi_{a+1,a} \\
     (-1)^\gr{1} \xi_{a+2,1} & (-1)^\gr{2} \xi_{a+2,2} & \cdots & (-1)^\gr{a} \xi_{a+2,a} \\
     \vdots & \vdots & \ddots & \vdots \\
     (-1)^\gr{1} \xi_{n+m,1} & (-1)^\gr{2} \xi_{n+m,2} & \cdots & (-1)^\gr{a} \xi_{n+m,a}
  \end{array}\right) .
\end{equation}
Here, $(\xi_{ij},\bar\xi_{k\ell})$ are superoscillators with $\BZ_2$-grading $\gr{\xi_{ij}}=\gr{i}+\gr{j}$
and $\gr{\bar\xi_{k\ell}}=\gr{k}+\gr{\ell}$ that obey the following commutation relations:
\begin{equation}
  [\xi_{ij},\bar\xi_{k\ell}] = \delta_{i\ell}\delta_{jk} \,, \qquad
  [\xi_{ij},\xi_{\ell k}] = 0 \,, \qquad
  [\bar\xi_{k\ell},\bar\xi_{ji}] = 0
\end{equation}
for any $1\leq j,k\leq a<i,\ell \leq n+m$, with the supercommutator $[-,-]$ defined in~\eqref{eq:super commutator}.

It was further noticed in \cite{Frasseksusy} that for any $0\leq a\leq n+m$, one has a total of $\binom{n+m}{a}$
Lax matrices analogous to~\eqref{eq:one matrices1}, see \cite[(2.12)]{Frasseksusy} onwards, corresponding to the
number of choices to distribute $a$ spectral parameters $x$ on the main diagonal of the Lax matrix of size $n+m$.
Evoking the invariance of the $R$-matrix~\eqref{eq:Atype-Rmatrix} under the symmetric group $S(n+m)$, see~\eqref{eq:Rmatrix-intertwiner},
these Lax matrices can be obtained by permuting rows and columns of the one in \eqref{eq:one matrices1}, see
Remark~\ref{rem:different parities}. The following family of such Lax matrices will be relevant to us in the following:
\begin{equation}\label{eq:two matricesb}
  \bL_a(y) = \bL^\VV_a(y) = \,
  \left(\begin{BMAT}[5pt]{c:c}{c:c}
    \ID_a & \kp \\
    \km & y\ID_{n+m-a} + \km\kp
  \end{BMAT}\right) \,.
\end{equation}

Let us first explain how these Lax matrices are related to those in \eqref{eq:one matrices1}.
To this end, let $\bar{\VV}$ be the superspace with a basis $\{\bar{\sfv}_i\}_{i=1}^{n+m}$
whose parity sequence~\eqref{eq:parity-sequence} is opposite to that of $\VV$, that is
\begin{equation}\label{eq:vvb}
  |\bar{\sfv}_i|=|\sfv_{n+m+1-i}| \qquad \forall\, 1\leq i\leq n+m \,.
\end{equation}
The corresponding rational $R$-matrices are related via~\eqref{eq:Rmatrix-intertwiner} with $\sigma(i)=n+m+1-i$,
so that
\begin{equation}\label{eq:idb-matrix}
  \idb_\sigma = \idb_{n+m} =
  \left(\begin{BMAT}(e)[0.8pt]{cccc}{cccc}
    0 & \cdots & 0 & 1 \\
    \vdots & \iddots & \iddots & 0 \\
    0 & \iddots & \iddots & \vdots \\
    1 & 0 & \cdots & 0
  \end{BMAT}\right) \,.
\end{equation}
Then, as noted in Remark~\ref{rem:different parities}, the matrix $\hat{L}^{\tilde{\VV}}_{n+m-a}(z)$
obtained from \eqref{eq:two matricesb} through
\begin{equation}\label{eq:hat-L}
  \hat{L}^{\bar{\VV}}_{n+m-a}(x) = \idb_{n+m} \bL^\VV_{a}(x) \idb_{n+m}^{-1}
\end{equation}
is Lax and has opposite $\BZ_2$-grading to that of $\bar{L}^\VV_{a}(x)$, i.e.\ the underlying vector superspaces
have opposite parity sequences. Using the notation $\bar{a}=n+m-a$ and $i'=n+m+1-i$, we find that
\begin{equation}
  \hat{L}^{\bar{\VV}}_{\bar a}(x) = \,
  \left(\begin{BMAT}[5pt]{c:c}{c:c}
    x\ID_{\bar a}+\qm\qp & \qm \\
    \qp & \ID_{n+m-\bar a}
  \end{BMAT}\right)
\end{equation}
with
\begin{equation}
  \qp =
  \left(\begin{array}{cccc}
    \bar\xi_{(\bar a+1)',1'} & \bar\xi_{(\bar a+1)',2'} & \cdots & \bar \xi_{(\bar a+1)',{\bar a}'} \\
    \bar\xi_{(\bar a+2)',1'} & \bar\xi_{(\bar a+2)',2'} & \cdots & \bar \xi_{(\bar a+2)',{\bar a}'} \\
    \vdots & \vdots & \ddots & \vdots \\
    \bar\xi_{(n+m)',1'} & \bar\xi_{(n+m)',2'} & \cdots & \bar \xi_{(n+m)',{\bar a}'}
  \end{array}\right)
\end{equation}
and
\begin{equation}
  \qm =
  \left(\begin{array}{cccc}
    (-1)^\gr{(\bar a+1)'}\xi_{1',(\bar a+1)'} & (-1)^\gr{(\bar a+2)'} \xi_{1',(\bar a+2)'} & \cdots &(-1)^\gr{(n+m)'} \xi_{1',(n+m)'} \\
    (-1)^\gr{(\bar a+1)'} \xi_{2',(\bar a+1)'} & (-1)^\gr{(\bar a+2)'} \xi_{2',(\bar a+2)'}  & \cdots & (-1)^\gr{(n+m)'} \xi_{2',(n+m)'} \\
    \vdots & \vdots & \ddots & \vdots \\
    (-1)^\gr{(\bar a+1)'} \xi_{{\bar a}',(\bar a +1)'} & (-1)^\gr{(\bar a+2)'} \xi_{{\bar a}',(\bar a+2)'}  & \cdots &(-1)^\gr{(n+m)'} \xi_{{\bar a}',(n+m)'}
  \end{array}\right) .
\end{equation}
Applying further the particle-hole transformation
\begin{equation}\label{eq:p.h.-needed}
  \bar \xi_{i'j'} \mapsto - (-1)^\gr{j'}\xi_{ij} \,,\qquad
  \xi_{j'i'} \mapsto (-1)^{\gr{i'}}\bar\xi_{ji}
  \qquad \forall \, 1\leq j\leq \bar{a}<i\leq n+m \,,
\end{equation}
we obtain the Lax matrix of \eqref{eq:one matrices1} defined on the vector space $\bar{\VV}$ with the opposite grading, i.e.
\begin{equation}\label{eq:one matrices3}
  L^{\bar{\VV}}_{\bar a}(x) = \hat{L}^{\bar{\VV}}_a(x)|_{p.h.} \,.
\end{equation}
It thus follows that $L^\VV_{a}(x)$ of~\eqref{eq:one matrices1} and $\bL^{\VV}_{a}(x)$ of~\eqref{eq:two matricesb}
satisfy the same RTT-relation \eqref{eq:rttA}.

\begin{Rem}
The Lax matrix \eqref{eq:two matricesb} is obtained through a particle-hole transformation~\eqref{eq:p.h.-needed}
from that of~\cite[(3.4)]{Frasseksusy} for $I=\{1,\ldots,a\}, J=\{a+1,\ldots,n+m\}$, the trivial representation
of the additional copy of $\gl(p|q)$, and an appropriate shift of the spectral parameter, cf.~Footnote~\ref{fn1}.
\end{Rem}

Let us now consider two copies of mutually supercommuting superoscillators
  $\big\{\big(\xi^{[r]}_{ij},\bar\xi^{[r]}_{ji}\big)\big\}_{1\leq j\leq a}^{a<i\leq n+m}$,
where the superscript $r=1,2$ indicates whether they appear in $L^{\VV,[1]}_a(x)$ or $\bL^{\VV,[2]}_a(y)$,
respectively. The subsequent factorisation was considered in~\cite[\S3.1]{Frasseksusy}:
\begin{equation}\label{eq:linear factorization1}
  L_a^{\VV,[1]}(x)\bL_a^{\VV,[2]}(y) = \,
    \left(\begin{BMAT}[5pt]{c:c}{c:c}
      x\ID_{a}-\kp'_1\km'_1 & \left((y-x)\ID_a+\kp'_1\km'_1\right)\kp'_1 \\
      - \km'_1 & y\ID_{n+m-a}+\km_1'\kp'_1
    \end{BMAT}\right)\,
    \left(\begin{BMAT}[5pt]{c:c}{c:c}
      \ID_{a} & \kp'_2\\
      0 & \ID_{n+m-a}
    \end{BMAT}\right)
\end{equation}
where
\begin{equation}\label{eq:tra121}
\begin{split}
  & \km'_1=\km_1-\km_2 \,,\qquad \kp'_1=\kp_1 \,, \\
  & \kp'_2=\kp_2+\kp_1 \,,\qquad \km'_2=\km_2 \,,
\end{split}
\end{equation}
and the subscript $r=1,2$ denotes the corresponding family of oscillators.
As noted in \cite[\S3.1]{Frasseksusy}, the generators~\eqref{eq:tra121} are related to those
in~\eqref{eq:one matrices1}--\eqref{eq:two matricesb} through a similarity~transformation:
\begin{equation}\label{eq:Atype-similarity}
  \km_r' = \mathbf{S}_a \km_r \mathbf{S}_a^{-1} \,,\qquad
  \kp_r' = \mathbf{S}_a \kp_r \mathbf{S}_a^{-1} \qquad (r=1,2)
\end{equation}
with
\begin{equation}\label{eq:S-Atype}
  \mathbf{S}_a = \exp\left[\sum_{1\leq i\leq a}^{a<j\leq n+m} \bar\xi_{ij}^{[1]} \xi_{ji}^{[2]}\right] \,.
\end{equation}
We note that all the summands in the exponent above are bosonic\footnote{Elements of a superalgebra are called
``bosonic'' or ``fermionic'' if their $\BZ_2$-degree is $\bar{0}$ or $\bar{1}$, respectively.} and pairwise supercommute.

\begin{Rem}
It immediately follows from~\eqref{eq:Atype-similarity}, but can be also directly checked from~\eqref{eq:tra121},
that the entries of the matrices $\km_r',\kp_r'$ encode mutually supercommuting pairs of superoscillators.
\end{Rem}

It follows that for any $x_1,x_2\in \BC$, the matrix
\begin{equation}\label{eq:ALax-non-degenerate}
  \mathcal{L}_{x_1,x_2}(x) = \mathcal{L}^\VV_{x_1,x_2}(x) =
  \left(\begin{BMAT}[5pt]{c:c}{c:c}
    (x+x_1)\ID_{a}-\kp_1\km_1 & \left((x_2-x_1)\ID_{a}+\kp_1\km_1\right)\kp_1 \\
     -\km_1 & (x+x_2)\ID_{n+m-a}+\km_1\kp_1
  \end{BMAT}\right)
\end{equation}
is a solution to the RTT-relation \eqref{eq:rttA}, hence, is Lax. Moreover, it arises through the fusion
\begin{equation}\label{eq:linear factorizationLax2}
  L_a^{\VV,[1]}(x+x_1)\bL_a^{\VV,[2]}(x+x_2)=
  \mathbf{S}_a \, \mathcal{L}^\VV_{x_1,x_2}(x) \,
  \left(\begin{BMAT}[5pt]{c:c}{c:c}
    \ID_{a} & \kp_2\\
    0 & \ID_{n+m-a}
  \end{BMAT}\right)
  \mathbf{S}_a^{-1}
\end{equation}
with the similarity transformation $\mathbf{S}_a$ of~\eqref{eq:S-Atype}.

\begin{Rem}\label{rem:Atype-normalized-limits}
Similarly to~\cite[\S8]{fkt}, we note that we can vice versa obtain the degenerate linear matrices $L_a^{\VV}(x)$
and $\bL_a^{\VV}(x)$ of~(\ref{eq:one matrices1},~\ref{eq:two matricesb}) from the non-degenerate linear Lax matrix
$\mathcal{L}^\VV_{x_1,x_2}(x)$ of~\eqref{eq:ALax-non-degenerate} via the
\emph{renormalized limit} procedures (which clearly preserve the property of being Lax):
\begin{equation}\label{eq:A renormalized limit}
\begin{split}
  & L^\VV_{a}(x) =
    \lim_{t\to \infty}\ \Big\{\mathcal{L}^\VV_{0,t}(x)\cdot \mathrm{diag}
    \Big(\underbrace{1,\ldots,1}_{a};\underbrace{\sfrac{1}{t},\ldots,\sfrac{1}{t}}_{n+m-a}\Big)\Big\} \,, \\
  & \bL_a^{\VV}(x) = \lim_{t\to \infty}\ \Big\{\mathrm{diag}
    \Big(\underbrace{\sfrac{1}{t},\ldots,\sfrac{1}{t}}_{a};\underbrace{1,\ldots,1}_{n+m-a}\Big)
    \cdot \mathcal{L}^\VV_{t,0}(x)\Big\}\Big|_{\bar\xi_{ij} \mapsto -\bar\xi_{ij} \,,\, \xi_{ij} \mapsto -\xi_{ij}}  \,.
\end{split}
\end{equation}
\end{Rem}


\section{Orthosymplectic Lax matrices}\label{sec:setup-and-invariance}

In this section, we set the notation for the orthosymplectic $R$-matrix and the corresponding Lax matrices,
as well as discuss the invariance of this $R$-matrix that will be instrumental later on.


\subsection{Orthosymplectic setup}
\

Fix $N,m\geq 0$, and consider the set $\mathbb{I}:=\{1,2,\ldots,N+2m\}$ equipped with an involution $'$:
\begin{equation}\label{eq:index-inv}
  i':=N+2m+1-i \,.
\end{equation}
Let $V$ be a superspace with a $\BZ_2$-homogeneous basis $v_1,\ldots,v_{N+2m}$ such that
\begin{equation}\label{eq:even+odd-dim}
  \dim(V_{\bar{0}})=N \,,\qquad \dim(V_{\bar{1}})=2m \,,
\end{equation}
and the grading is $\BZ_2$-symmetric in the following sense:
\begin{equation}\label{eq:grading-symmetry}
  |v_i|=|v_{i'}| \qquad \forall \, 1\leq i\leq N+2m \,.
\end{equation}
For the major part of our constructions (except for Subsections~\ref{ssec:odd case}--\ref{ss:facq}), we shall
assume that $N$ is even: $N=2n$. In this case, we pick the following specific $\BZ_2$-grading of $V$:
\begin{equation}\label{eq:our-osp-grading}
  \gr{i}:=|v_i|=
  \begin{cases}
    \bar{0} & \quad \text{for} \ 1\leq i\leq n \\
    \bar{1} & \quad \text{for} \ n+1\leq i\leq n+2m \\
    \bar{0} & \quad \text{for} \ n+2m+1\leq i\leq 2n+2m
   \end{cases}
\end{equation}
which corresponds to the following parity sequence, cf.~\eqref{eq:parity-sequence}:
\begin{equation}\label{eq:our-parity}
  \Parity =
  \Big(
    \underbrace{\bar{0},\ldots,\bar{0}}_{n}, \underbrace{\bar{1},\ldots,\bar{1}}_{2m},\underbrace{\bar{0},\ldots,\bar{0}}_{n}
  \Big) \,.
\end{equation}

Similarly to~\eqref{eq:P}, we consider the permutation operator $\Pop\colon V\otimes V \to V \otimes V$ defined by
\begin{equation}\label{eq:Pop}
  \Pop = \sum_{i,j=1}^{N+2m}(-1)^{\gr{j}}\, e_{ij}\otimes e_{ji} \,.
\end{equation}
We also consider the operator $\Qop \colon V\otimes V\to V\otimes V$ defined by
\begin{equation}\label{eq:Qop}
  \Qop = \sum_{i,j=1}^{N+2m} (-1)^{\gr{i}\gr{j}} \theta_i\theta_j \,  e_{ij}\otimes e_{i'j'} \,.
\end{equation}
Here, the sequence $\theta=\theta_V=(\theta_1,\ldots,\theta_{N+2m})$ of $\pm 1$'s is determined uniquely by the conditions
\begin{equation}\label{eq:grrel}
  \theta_{i'}=(-1)^\gr{i}\theta_i \,,\qquad \theta_{\leq (n+m)}=1 \,,
\end{equation}
so that
\begin{equation}\label{eq:theta}
  \theta=\theta_V=\Big(\underbrace{1,\ldots,1}_{n+m},\underbrace{-1,\ldots,-1}_{m},\underbrace{1,\ldots,1}_{n}\Big)
\end{equation}
for the specific $\BZ_2$-grading~\eqref{eq:our-osp-grading}. Explicitly, the action of $\Qop$ is given by:
\begin{equation}\label{eq:Qop-explicit}
  \Qop(v_a\otimes v_b)=
  \begin{cases}
    0 & \mbox{if } b\ne a' \\
    \sum_{i=1}^{N+2m} \theta_i \, v_i\otimes v_{i'} & \mbox{if } b=a' \,,\, a>n+m \\
    (-1)^{|a|}\sum_{i=1}^{N+2m} \theta_i\, v_i\otimes v_{i'} & \mbox{if } b=a' \,,\, a\leq n+m
  \end{cases} \,.
\end{equation}
We also introduce a constant $\kappa$ via:
\begin{equation}\label{eq:kappa}
  \kappa=\frac{N}{2}-m-1=n-m-1 \,.
\end{equation}
Consider the \emph{rational} $R$-matrix (a super-version of the one considered in~\cite{zz}):
\begin{equation}\label{eq:R-osp}
  R(x)=R^V(x)=x(x+\kappa)\ID + (x+\kappa)\Pop - x\Qop \,,
\end{equation}
which satisfies the Yang-Baxter equation with a spectral parameter~\eqref{eq:YBE} according to~\cite{Arnaudon}.

For any superalgebra $A$, an even matrix $L(x)=L^V(x)=(L_{ij}(x))_{i,j=1}^{N+2m}\in \End\,V\otimes A[[x,x^{-1}]]$
will be called an (orthosymplectic) \emph{Lax matrix} if it satisfies the RTT-relation~\eqref{eq:rttA} with $R(x)$
of~\eqref{eq:R-osp}. Coefficient-wise, this is equivalent to the following system of relations (see~\cite{Arnaudon}):
\begin{equation}\label{eq:yangianalg}
\begin{split}
  [L_{ij}(x),L_{k\ell}(y)]
  &=\frac{(-1)^{\gr{i}\gr{j}+\gr{i}\gr{k}+\gr{j}\gr{k}}}{x-y}
       \Big(L_{kj}(y)L_{i\ell}(x)-L_{kj}(x)L_{i\ell}(y)\Big) \, + \\
  & \frac{1}{x-y+\kappa} \left(\delta_{ki'}\sum_{p=1}^{N+2m} L_{pj}(x)L_{p'\ell}(y)
    (-1)^{\gr{i}+\gr{i}\gr{j}+\gr{j}\gr{p}}\theta_{i}\theta_{p} \, - \right.\\
  & \left.\qquad\qquad\quad \delta_{\ell j'}\sum_{p=1}^{N+2m} L_{kp'}(y)L_{ip}(x)
     (-1)^{\gr{p}+\gr{j}+\gr{i}\gr{k}+\gr{i}\gr{p}+\gr{j}\gr{k}}\theta_{p}\theta_{j}\right)\,.
\end{split}
\end{equation}

\begin{Rem}
The data of a Lax matrix $(L_{ij}(x))_{i,j=1}^{N+2m}$ with $L_{ij}(x)\in \delta_{ij}+x^{-1}A[[x^{-1}]]$ is
equivalent to an algebra homomorphism $X^\rtt(\fosp(V))\to A$ from an RTT extended orthosymplectic Yangian.
\end{Rem}

\begin{Rem}\label{rem:CD-vs-osp}
We recover orthogonal and symplectic types as special cases of the above setup:
\begin{itemize}

\item
For $m=0$, we have $\theta=(\underbrace{1,\ldots,1}_{2n})$ and $\gr{i}=\bar{0}$ for all $1\leq i\leq 2n$,
so that $R(x)$ of~\eqref{eq:R-osp} coincides with the $D_{n}$-type rational $R$-matrix of~\cite[(1.19)]{fkt}.

\item
For $n=0$, we have $\theta=(\underbrace{1,\ldots,1}_{m},\underbrace{-1,\ldots,-1}_{m})$ and $\gr{i}=\bar{1}$ for
all $1\leq i\leq 2m$, so that our $R(-x)$ of~\eqref{eq:R-osp} coincides with the $C_{m}$-type rational $R$-matrix
of~\cite[(1.19)]{fkt}.

\end{itemize}
\end{Rem}

\begin{Rem}\label{rem:different parities osp}
Let $\tilde{V}$ be another superspace with a $\BC$-basis $\tilde{v}_1,\ldots,\tilde{v}_{N+2m}$
satisfying~(\ref{eq:even+odd-dim},~\ref{eq:grading-symmetry}). Pick a permutation $\sigma\in S(n+m)$
such that $v_i\in V$ and $\tilde{v}_{\sigma(i)}\in \tilde{V}$ have the same $\BZ_2$-grading for all
$1\leq i\leq n+m$, and extend it to $\sigma\in S(N+2m)$ via $\sigma(i')=\sigma(i)'$. Define a
superspace isomorphism $\tilde{\idb}_\sigma\colon V \iso \tilde{V}$ via $v_i\mapsto \tilde{v}_{\sigma(i)}$.
The corresponding $R$-matrices~\eqref{eq:R-osp} are related via
\begin{equation}\label{eq:Rmatrix-intertwiner-osp}
  R^{\tilde{V}}(x) =
  (\tilde{\idb}_\sigma\otimes \tilde{\idb}_\sigma)\, R^V(x)\, (\tilde{\idb}_\sigma\otimes \tilde{\idb}_\sigma)^{-1} \,.
\end{equation}
As a result, if $L^V(x)=(L_{ij}(x))_{i,j=1}^{N+2m}$ is a Lax matrix corresponding to the fixed parity
$\Upsilon_V$, then $L^{\tilde{V}}(x):=\tilde{\idb}_\sigma L^V(x) \tilde{\idb}_\sigma^{-1}$ is a Lax matrix
corresponding to the parity $\Upsilon_{\tilde{V}}$, cf.~Remark~\ref{rem:different parities}.
\end{Rem}


\subsection{Symmetries of the orthosymplectic $R$-matrix}
\

In this section, we establish the invariance of the $R$-matrix~\eqref{eq:R-osp} under certain \emph{graded permutation}
matrices that will be used in the later constructions. The proofs are based on direct computations of the commutators of
the corresponding operators with $\Pop,\Qop$ of~(\ref{eq:Pop},~\ref{eq:Qop}). To this end, let us recall the explicit formulas
for the action of $\Pop$ and $\Qop$, cf.~(\ref{eq:Pop-explicit},~\ref{eq:Qop-explicit}):
\begin{equation}
\begin{split}
  & \Pop \colon v_i\otimes v_j\mapsto (-1)^{\gr{i}\gr{j}} \, v_j\otimes v_i \,, \\
  & \Qop\colon v_{i}\otimes v_{i'}\mapsto (-1)^{\gr{i}}\theta_{i} \sum_{j=1}^{N+2m} \theta_j\, v_j\otimes v_{j'}
    \,, \quad v_{i}\otimes v_{\jmath\ne i'}\mapsto 0 \,.
\end{split}
\end{equation}
Henceforth, $\JD_r$ will denote the $r\times r$ matrix with ``1'' on the antidiagonal, cf.~\eqref{eq:idb-matrix}.

\begin{Lem}\label{lem:jjt}
The $R$-matrix \eqref{eq:R-osp} commutes with the tensor product of two matrices
\begin{equation}\label{eq:jd-theta}
  \JD_\theta=
  \left(\begin{BMAT}[5pt]{c:c}{c:c}
    0 & -\JD_{n} \\
    \JD_{n+2m } & 0
  \end{BMAT} \right)
\end{equation}
that is, $[R(x),\JD_\theta\otimes \JD_\theta]=0$.
\end{Lem}

\begin{proof}
The tensor product of the matrices $\JD_\theta$ acts explicitly via
\begin{equation}
  \JD_\theta\otimes \JD_\theta\colon
  v_i\otimes v_j \mapsto (-1)^{\delta_{i>n+2m}+\delta_{j>n+2m}}\, v_{i'}\otimes v_{j'} \,.
\end{equation}
We thus have:
\begin{align}
  \Pop\left(\idb_\theta\otimes \idb_\theta\right) & \colon v_i\otimes v_j \mapsto
    (-1)^{\gr{i'}\gr{j'}+\delta_{i>n+2m}+\delta_{j>n+2m}}\, v_{j'}\otimes v_{i'} \,,  \\
  \left(\idb_\theta\otimes \idb_\theta\right)\Pop & \colon v_i\otimes v_j \mapsto
    (-1)^{\gr{i}\gr{j}+\delta_{j>n+2m}+\delta_{i>n+2m}}\, v_{j'}\otimes v_{i'} \,,
\end{align}
so that $[\Pop,\idb_\theta\otimes \idb_\theta]=0$ since $\gr{\iota}=\gr{\iota'}$ for all $\iota$.

To check the invariance of $\Qop$, we first note that
  $$\Qop\left(\idb_\theta\otimes \idb_\theta\right)(v_i\otimes v_\jmath)=0=
    \left(\idb_\theta\otimes \idb_\theta\right)\Qop\, (v_i\otimes v_\jmath)
    \qquad \mathrm{for} \quad \jmath\ne i' \,.$$
It thus remains to compare the images of $v_i\otimes v_{i'}$ under both
$\Qop\left(\idb_\theta\otimes \idb_\theta\right)$ and $\left(\idb_\theta\otimes \idb_\theta\right)\Qop$:
\begin{align}
  \Qop\left(\idb_\theta\otimes \idb_\theta\right) &\colon v_i\otimes v_{i'}\mapsto
    (-1)^{\gr{i'}+\delta_{i>n+2m}+\delta_{i'>n+2m}}\theta_{i'}
    \sum_{j=1}^{N+2m}\theta_j \, v_{j}\otimes v_{j'} \,,
  \label{eq:infJ1} \\
  \left(\idb_\theta\otimes \idb_\theta\right)\Qop &\colon v_i\otimes v_{i'}\mapsto
    (-1)^{\gr{i}}\theta_{i}
    \sum_{j=1}^{N+2m}(-1)^{\delta_{j>n+2m}+\delta_{j'>n+2m}}\theta_j \, v_{j'}\otimes v_{j} \,.
  \label{eq:infJ2}
\end{align}
The right-hand sides of~(\ref{eq:infJ1},~\ref{eq:infJ2}) coincide due to
the first equality in~\eqref{eq:grrel} and the identity
  $$(-1)^{\delta_{i>n+2m}+\delta_{i'>n+2m}}=(-1)^{\delta_{i>n+2m}+\delta_{i\leq n}}=-(-1)^\gr{i}$$
that follows immediately from~\eqref{eq:our-osp-grading}.
This verifies $[\Qop,\idb_\theta\otimes \idb_\theta]=0$.

As $\idb_\theta\otimes\idb_\theta$ commutes with both $\Pop$ and $\Qop$, so it does with the $R$-matrix \eqref{eq:R-osp}.
\end{proof}

\begin{Lem}\label{lem:tj}
The $R$-matrix \eqref{eq:R-osp} commutes with the tensor product of two matrices
\begin{equation}\label{eq:tilde-idb-matrix}
  \tilde \idb=
  \left(\begin{BMAT}[5pt]{c|c|c}{c|c|c}
    0 & 0 & 1 \\
    0 & \id_{N+2m-2} & 0\\
    1 & 0 & 0
  \end{BMAT} \right)
\end{equation}
that is, $[R(x),\tilde\JD\otimes \tilde\JD]=0$.
\end{Lem}

\begin{proof}
The tensor product of the matrices $\tilde\JD$ acts explicitly via
\begin{equation}
  \tilde\JD \otimes \tilde\JD \colon v_i\otimes v_j \mapsto v_{\tilde i}\otimes v_{\tilde j}
    \qquad  \text{with} \quad
  \tilde i=
  \begin{cases}
    i' & \mathrm{if} \ i=1,1' \\
    i  & \mathrm{if} \ i=2,\ldots, 2'
  \end{cases} \,.
\end{equation}
We thus have:
\begin{align}
  \Pop\big(\tilde\idb\otimes\tilde \idb\big) & \colon v_i\otimes v_j \mapsto
    (-1)^{\gr{\tilde i}\gr{\tilde j}} \, v_{\tilde j}\otimes v_{\tilde i} \,,  \\
  \big(\tilde\idb\otimes\tilde \idb\big)\Pop & \colon v_i\otimes v_j\mapsto
    (-1)^{\gr{i}\gr{j}} \, v_{\tilde j}\otimes v_{\tilde i} \,,
\end{align}
so that $[\Pop,\tilde\idb\otimes\tilde \idb]=0$ since $\gr{\tilde{\iota}}=\gr{\iota}$ for all $\iota$.

To check the invariance of $\Qop$, we first note that
  $$\Qop\big(\tilde\idb\otimes\tilde \idb\big)(v_i\otimes v_\jmath)=0=
    \big(\tilde\idb\otimes\tilde \idb\big)\Qop\, (v_i\otimes v_\jmath)
    \qquad \mathrm{for} \quad \jmath\ne i' \,.$$
It thus remains to compare the images of $v_i\otimes v_{i'}$ under both
$\Qop\big(\tilde\idb\otimes\tilde \idb\big)$ and $\big(\tilde\idb\otimes\tilde \idb\big)\Qop$:
\begin{align}
  \Qop \big(\tilde\idb\otimes\tilde \idb\big) & \colon v_i\otimes v_{i'}\mapsto
    (-1)^{\gr{\tilde i}}\theta_{\tilde i} \sum_{j=1}^{N+2m}\theta_j \, v_{j}\otimes v_{j'} \,, \\
  \big(\tilde\idb\otimes\tilde \idb\big)\Qop &\colon v_i\otimes v_{i'}\mapsto
    (-1)^{\gr{i}}\theta_{i}\sum_{j=1}^{N+2m}\theta_j\, v_{\tilde j}\otimes v_{\tilde j'} \,,
\end{align}
and the two images coincides as $\theta_{\tilde i}=\theta_{i}$ (which uses that $|v_1|=\bar{0}$).

As $\tilde \idb \otimes \tilde \idb$ commutes with both $\Pop$ and $\Qop$, so it does with the $R$-matrix \eqref{eq:R-osp}.
\end{proof}

\begin{Lem}\label{lem:hjt}
The $R$-matrix \eqref{eq:R-osp} commutes with the tensor product of two matrices
\begin{equation}\label{eq:hatJd-idg}
  \hat \JD_\theta=
  \left(\begin{BMAT}[5pt]{c|c:c|c}{c|c:c|c}
    1 & 0 & 0 & 0 \\
    0 & 0 & \idg_{n-1,m} & 0 \\
    0 & \JD_{n+m-1} & 0 & 0 \\
    0 & 0 & 0 & -1
  \end{BMAT} \right)
    \quad \text{with} \quad
  \idg_{n-1,m} =
  \left(\begin{BMAT}[5pt]{c:c}{c:c}
    0 & -\JD_{n-1} \\
    \JD_m & 0
  \end{BMAT}\right)
\end{equation}
that is, $[R(x),\hat \JD_\theta \otimes \hat \JD_\theta]=0$.
\end{Lem}

\begin{proof}
As $\hat \idb_\theta=\tilde \idb\cdot \idb_\theta$, the result follows immediately from the previous two lemmas.
\end{proof}

\begin{Lem}\label{lem:jt}
For $N=0$, the $R$-matrix \eqref{eq:R-osp} commutes with the tensor product of two matrices
\begin{equation}\label{eq:ID-theta matrix}
  \ID_\theta=
  \left(\begin{BMAT}[5pt]{c:c}{c:c}
    \id_m & 0 \\
    0 & -\id_m
  \end{BMAT} \right)
\end{equation}
that is, $[R(x),\id_\theta \otimes \id_\theta]=0$.
\end{Lem}

\begin{proof}
For $N=0$, the $\theta$ of~\eqref{eq:theta} are given by $\theta_i=(-1)^{\delta_{i>m}}$, so that
\begin{equation}
  \id_\theta \otimes \id_\theta\colon v_i\otimes v_j\mapsto \theta_i\theta_j \, v_i\otimes v_j \,.
\end{equation}
The commutativity $[\Pop,\id_\theta\otimes \id_\theta]=0$ follows immediately from
\begin{align}
  \Pop\left( \id_\theta \otimes \id_\theta\right) & \colon v_i\otimes v_j\mapsto
    (-1)^{\gr{i}\gr{j}} \theta_i\theta_j \, v_j\otimes v_i \,, \\
  \left( \id_\theta \otimes \id_\theta\right)\Pop & \colon v_i\otimes v_j\mapsto
    (-1)^{\gr{i}\gr{j}} \theta_i\theta_j \, v_j\otimes v_i \,.
\end{align}
Similarly, for the operator $\Qop$ we have:
\begin{align}
  \Qop\left( \id_\theta \otimes \id_\theta\right) & \colon v_i\otimes v_{i'} \mapsto
    (-1)^{\gr{i}} \theta_{i'}\sum_{j=1}^{2m} \theta_j \, v_j\otimes v_{j'} \,, \\
  \left( \id_\theta \otimes \id_\theta\right)\Qop & \colon v_i\otimes v_{i'}\mapsto
    (-1)^{\gr{i}} \theta_i\sum_{j=1}^{2m} \theta_{j'} \, v_j\otimes v_{j'} \,,
\end{align}
and the two images coincide as $\theta_\iota=-\theta_{\iota'}$ for all $\iota$.
This implies $[\Qop,\id_\theta\otimes \id_\theta]=0$.

As $\id_\theta \otimes \id_\theta$ commutes with both $\Pop$ and $\Qop$, so it does with the $R$-matrix \eqref{eq:R-osp}.
\end{proof}

\begin{Rem}
(a) More generally, the $R$-matrix \eqref{eq:R-osp} commutes with
$\sum_{i=1}^{n+m} (a_ie_{ii}+b_ie_{ii'}+c_ie_{i'i}+d_ie_{i'i'})$ with
either $a_i,d_i\in \{\pm 1\}$ and $b_i=c_i=0$ (for which we set $\gamma_i:=a_id_i$) or
$b_i,c_i\in \{\pm 1\}$ and $a_i=d_i=0$ (for which we set $\gamma_i:=(-1)^{|i|}b_ic_i$),
and such that $\gamma_i$ are the same for all $1\leq i\leq n+m$.

\medskip
\noindent
(b) According to Remark~\ref{rem:different parities osp}, using some other permutation matrices ($\tilde{\idb}_\sigma$
from \emph{loc.cit.}) will rather produce orthosymplectic Lax matrices for other $\BZ_2$-gradings of $V$, see
Remarks~\ref{rem:linear-Lax-all-parities} and~\ref{rem:quadratic-Lax-all-parities}.
\end{Rem}


\section{Linear orthosymplectic Lax matrices}\label{sec:linear-Lax}

In this section, we construct some linear orthosymplectic Lax matrices of superoscillator type.


\subsection{Degenerate linear orthosymplectic Lax matrices}
\

In this subsection, we construct a degenerate linear orthosymplectic Lax matrix
for the parity sequence~\eqref{eq:our-parity}. To this end, let us consider first bosonic pairs of superoscillators:
\begin{equation}\label{eq:aoscs}
  \Big\{(\oa_{ij},\oad_{ji}) \,\Big|\, n+2m+1\leq i\leq 2n+2m-1 \,,\ 1\leq j\leq n-1 \,,\ i+j\leq 2n+2m \Big\} \,,
\end{equation}
\begin{equation}\label{eq:boscs}
  \Big\{(\ob_{ij},\obd_{ji}) \,\Big|\, n+m+1\leq i\leq n+2m \,,\ n+1\leq j\leq n+m \,,\ i+j\leq 2n+2m+1 \Big\} \,,
\end{equation}
with $|\oa_{ij}|=|\oad_{ji}|=|i|+|j|=\bar{0}$, $|\ob_{ij}|=|\obd_{ji}|=|i|+|j|=\bar{0}$, and the only nonzero supercommutators
\begin{equation}
  [\oa_{ij},\oad_{ji}]=\oa_{ij}\oad_{ji}-\oad_{ji}\oa_{ij}=1 \,, \qquad
  [\ob_{ij},\obd_{ji}]=\ob_{ij}\obd_{ji}-\obd_{ji}\ob_{ij}=1 \,.
\end{equation}
In addition, we also consider fermionic pairs of superoscillators
\begin{equation}\label{eq:coscs}
  \Big\{(\oc_{ij},\ocd_{ji}) \,\Big|\,  n+m+1\leq i\leq n+2m \,,\ 1\leq j\leq n \Big\}
\end{equation}
with $|\oc_{ij}|=|\ocd_{ji}|=|i|+|j|=\bar{1}$ and the only nonzero supercommutators
\begin{equation}
  [\oc_{ij},\ocd_{ji}]=\oc_{ij}\ocd_{ji}+\ocd_{ji}\oc_{ij}=1 \,.
\end{equation}
We encode the above bosonic and fermionic generators by the corresponding three pairs of matrices:
\begin{equation}\label{eq:ApAmDp}
  \ap=
  \left(\begin{array}{cccc}
    \oad_{1,n+2m+1} & \cdots & \oad_{1,2n+2m-1} & 0 \\
    \vdots & \iddots & 0 & -\oad_{1,2n+2m-1}\\
    \oad_{n-1,n+2m+1} & 0 & \iddots & \vdots \\
    0 & -\oad_{n-1,n+2m+1} & \cdots & -\oad_{1,n+2m+1}
   \end{array}\right) \,,
\end{equation}
\begin{equation}\label{eq:ApAmDm}
  \am=
  \left(\begin{array}{cccc}
    \oa_{n+2m+1,1} & \cdots & \oa_{n+2m+1,n-1} & 0 \\
    \vdots & \iddots & 0 & -\oa_{n+2m+1,n-1} \\
    \oa_{2n+2m-1,1} & 0 & \iddots & \vdots \\
    0 & -\oa_{2n+2m-1,1} & \cdots & -\oa_{n+2m+1,1}
  \end{array}\right)
\end{equation}
that are (skew-symmetric along the antidiagonal) $n\times n$ matrices encoding~\eqref{eq:aoscs},
\begin{equation}\label{eq:ApAmCp}
  \bp=
  \left(\begin{array}{cccc}
    \obd_{n+1,n+m+1} & \cdots & \obd_{n+1,n+2m-1} & 2\obd_{n+1,n+2m} \\
    \vdots & \iddots & 2\obd_{n+2,n+2m-1} & \obd_{n+1,n+2m-1} \\
    \obd_{n+m-1,n+m+1} & 2\obd_{n+m-1,n+m+2} & \iddots & \vdots \\
    2\obd_{n+m,n+m+1} & \obd_{n+m-1,n+m+1} & \cdots & \obd_{n+1,n+m+1}
  \end{array}\right) \,,
\end{equation}
\begin{equation}\label{eq:ApAmCm}
  \bm=
  \left(\begin{array}{cccc}
    \ob_{n+m+1,n+1} & \cdots & \ob_{n+m+1,n+m-1} & \ob_{n+m+1,n+m} \\
    \vdots & \iddots & \ob_{n+m+2,n+m-1} & \ob_{n+m+1,n+m-1} \\
    \ob_{n+2m-1,n+1} & \ob_{n+2m-1,n+2} & \iddots & \vdots \\
    \ob_{n+2m,n+1} & \ob_{n+2m-1,n+1} & \cdots & \ob_{n+m+1,n+1}
  \end{array}\right)
\end{equation}
that are (symmetric along the antidiagonal) $m\times m$ matrices encoding~\eqref{eq:boscs},
\begin{equation}\label{eq:ApAmSUSY}
  \bar{\mathbf{C}}=
  \left(\begin{array}{ccc}
    \ocd_{1,n+m+1} & \cdots & \ocd_{1,n+2m}  \\
    \vdots & \ddots & \vdots \\
    \ocd_{n,n+m+1}  & \cdots & \ocd_{n,n+2m} \\
  \end{array}\right)
    \,, \qquad
  {\mathbf{C}}=
  \left(\begin{array}{ccc}
    \oc_{n+m+1,1} & \cdots & \oc_{n+m+1,n}  \\
    \vdots & \ddots & \vdots \\
    \oc_{n+2m,1}  & \cdots & \oc_{n+2m,n} \\
  \end{array}\right)
\end{equation}
that are $n\times m$ and $m\times n$ matrices encoding~\eqref{eq:coscs}.

Now we are ready to present the main result of this subsection:

\begin{Thm}\label{thm:degenerate-linear-Lax}
The following is a solution to the RTT-relation \eqref{eq:rttA} with the $R$-matrix~\eqref{eq:R-osp}:
\begin{equation}\label{eq:linlaxdeg}
  L(x)=
  \left(\begin{BMAT}[5pt]{c:c}{c:c}
    x\ID_{n+m}-\kp\km & \kp \\
    -\km & \ID _{n+m}
  \end{BMAT} \right)
\end{equation}
with
\begin{equation}\label{eq:Koscillators}
  \kp=
  \left(\begin{BMAT}[5pt]{c:c}{c:c}
    \bar{\mathbf{C}} & \ap\\
    \bp & -\JD_m \bar{\mathbf{C}}^t \JD_n
  \end{BMAT} \right)
    \qquad \mathrm{and} \qquad
  \km=
  \left(\begin{BMAT}[5pt]{c:c}{c:c}
    {\mathbf{C}} & -\bm \\
    \am & \JD_n {\mathbf{C}}^t \JD_m
  \end{BMAT} \right) \,.
\end{equation}
Here, $t$ denotes the standard (bosonic) transpose, $\JD_n$ is the $n\times n$ matrix with ``1'' on the antidiagonal,
and the matrices $\ap,\am,\bp,\bm,\bar{\mathbf{C}},\mathbf{C}$ are as in~\eqref{eq:ApAmDp}--\eqref{eq:ApAmSUSY}.
\end{Thm}

\begin{proof}
The proof is straightforward and is analogous to that for $C$- and $D$-types, directly verifying the commutation relations
using Lemma~\ref{lem:quadraticterm} below. In the first step, we insert the relations \eqref{eq:quadsimp1} and
\eqref{eq:quadsimp2} into \eqref{eq:yangianalg} and split the resulting equations into $4\times 4$ block structure according
to the block structure of the Lax matrices \eqref{eq:linlaxdeg}. The supercommutator on the left-hand side of
\eqref{eq:yangianalg} indicates that the right-hand side is symmetric under the combined exchange of $i \leftrightarrow k$,
$j \leftrightarrow \ell$ and $x\leftrightarrow y$ when multiplied by the factor $-(-1)^{(\gr{i}+\gr{j})(\gr{k}+\gr{\ell})}$.
This symmetry also holds for both terms individually on the right-hand side. For the first term, proportional to $(x-y)^{-1}$,
the symmetry follows from the relation
\begin{equation}
  L_{kj}(y)L_{i\ell}(x) - L_{i\ell}(x) L_{kj}(y) = L_{kj}(x)L_{i\ell}(y)-L_{i\ell}(y) L_{kj}(x) \,.
\end{equation}
For the second term, proportional to $(x-y+\kappa)^{-1}$, the symmetry is less obvious. It follows by noting that
\begin{equation}
\begin{split}
  \frac{1}{x-y+\kappa} \sum_{p}L_{pj}(x)L_{p'\ell}(y)(-1)^{\gr{j}\gr{p}}\theta_{p}
  & + \frac{(-1)^{\gr{\ell}\gr{j}}}{y-x+\kappa} \sum_{p}L_{p\ell}(y)L_{p'j}(x)(-1)^{\gr{\ell}\gr{p}}\theta_{p} \, = \\
  & \delta_{j\ell'}\theta_{j'}\frac{\kappa(x+y+\kappa)}{(x-y+\kappa)(y-x+\kappa)}
\end{split}
\end{equation}
and
\begin{equation}
\begin{split}
  \frac{1}{x-y+\kappa} \sum_{p}L_{kp'}(y)L_{ip}(x)(-1)^{\gr{i}\gr{p}}\theta_{p'}
  & + \frac{(-1)^{\gr{i}\gr{k}}}{y-x+\kappa} \sum_{p}L_{ip'}(x)L_{kp}(y)(-1)^{\gr{k}\gr{p}}\theta_{p'} \, = \\
  & \delta_{i'k}\theta_{i}\frac{\kappa(x+y+\kappa)}{(x-y+\kappa)(y-x+\kappa)} \,.
\end{split}
\end{equation}
These equations are verified using Lemma~\ref{lem:quadraticterm}. Thus, the $4\times4=16$ equations arising from
the block structure are reduced to $10$ equations. In the following, we verify the remaining commutation relations.

We start with the ``diagonal terms'' where the indices are chosen such that the commutator on the left-hand side is among
the elements of the Lax matrices within the same block. In this case, the second term on the right-hand side always
vanishes because of the Kronecker deltas. Further, unless we are considering the upper left block the first term
on the right-hand side has to vanish as the entries of the Lax matrix will not depend on the spectral parameter.
The case with $1\leq i,j,k,\ell\leq n+m$ follows from the following equality:
\begin{equation}
  \left[(\kp\km)_{ij},(\kp\km)_{k\ell}\right] =
  (-1)^{\gr{i}\gr{j}+\gr{i}\gr{k}+\gr{j}\gr{k}} \left(\delta_{kj}(\kp\km)_{i\ell}-\delta_{i\ell}(\kp\km)_{kj}\right) \,.
\end{equation}
This equality follows in turn from
\begin{equation}\label{eq:comK}
\begin{split}
  \km_{n+m-i+1,j}\kp_{k,n+m-\ell+1} - (-1)^{(\gr{i}+\gr{j})(\gr{k}+\gr{\ell})}
  & \kp_{k,n+m-\ell+1}\km_{n+m-i+1,j} \, = \\
  & (-1)^{\gr{j}} \delta_{i\ell}\delta_{jk} - (-1)^{\gr{i}\gr{j}+\gr{j}} \delta_{ik}\delta_{j\ell}
\end{split}
\end{equation}
for $1\leq i,j,k,\ell\leq n+m$, or equivalently (written in terms of the Lax matrix)
\begin{equation}\label{eq:comL}
  \left[L(x)_{i'j},L(y)_{k\ell'}\right] =
  (-1)^{\gr{i}\gr{j}+\gr{j}} \delta_{ik}\delta_{j\ell} - (-1)^{\gr{j}} \delta_{i\ell}\delta_{jk} \,.
\end{equation}
This leaves us with 6 equations to verify. The case
\begin{equation}\label{eq:yangianalg33}
\begin{split}
  [L_{ij}(x),L_{k\ell'}(y)]
  & = \frac{1}{x-y}(-1)^{\gr{i}\gr{j}+\gr{i}\gr{k}+\gr{j}\gr{k}} \Big(L_{kj}(y)L_{i\ell'}(x)-L_{kj}(x)L_{i\ell'}(y)\Big) \, - \\
  & \quad \frac{1}{x-y+\kappa} \theta_{j}(-1)^{\gr{j}+\gr{i}\gr{k}+\gr{j}\gr{k}}
    \delta_{\ell j}\sum_{p}L_{kp'}(y)L_{ip}(x) (-1)^{\gr{p}+\gr{i}\gr{p}}\theta_{p}
\end{split}
\end{equation}
with $1\leq i,j,k,\ell\leq n+m$ is reduced to proving
\begin{equation}\label{eq:yangianalg33b}
\begin{split}
  [(\kp\km)_{ij},\kp_{k,n+m-\ell+1}]
  & = (-1)^{\gr{i}\gr{j}+\gr{i}\gr{k}+\gr{j}\gr{k}} \delta_{kj}\, \kp_{i,n+m-\ell+1} +
      (-1)^{(\gr{i}+\gr{j})\gr{k}} \delta_{\ell j'}\, \kp_{k,n+m-i+1} \,.
\end{split}
\end{equation}
Similarly, the case
\begin{equation}\label{eq:yangianalg44}
\begin{split}
  [L_{ij}(x),L_{k'\ell}(y)]
  & = \frac{1}{x-y}(-1)^{\gr{i}\gr{j}+\gr{i}\gr{k}+\gr{j}\gr{k}} \Big(L_{k'j}(y)L_{i\ell}(x)-L_{k'j}(x)L_{i\ell}(y)\Big) \, + \\
  & \quad \frac{1}{x-y+\kappa} (-1)^{\gr{i}+\gr{i}\gr{j}} \theta_{i}\delta_{ki}
    \sum_{p}L_{pj}(x)L_{p'\ell}(y)(-1)^{\gr{j}\gr{p}}\theta_{p}
\end{split}
\end{equation}
with $1\leq i,j,k,\ell\leq n+m$ is reduced to proving
\begin{equation} \label{eq:kkcom2}
\begin{split}
  [(\kp\km)_{ij},\km_{n+m-k+1,\ell}]
  & = -(-1)^{\gr{i}\gr{j}+\gr{i}\gr{k}+\gr{j}\gr{k}} \delta_{i\ell}\, \km_{n+m-k+1,j} -
       (-1)^{\gr{j}+\gr{i}+\gr{i}\gr{j}} \delta_{ki}\, \km_{n+m-j+1,\ell} \,.
\end{split}
\end{equation}
Both equations \eqref{eq:yangianalg33b} and \eqref{eq:kkcom2} can be derived directly from \eqref{eq:comK}.

The case with $i,j',k',\ell$ for $1\leq i,j,k,\ell\leq n+m$ is verified using \eqref{eq:comL} and
Lemma~\ref{lem:quadraticterm}, while the remaining three equations follow directly from Lemma~\ref{lem:quadraticterm}.

This completes our proof of the theorem.
\end{proof}

\begin{Lem}\label{lem:quadraticterm}
 For the matrix $L(x)$ of~\eqref{eq:linlaxdeg}, we have the following matrix equalities:
\begin{equation}\label{eq:quadsimp1}
\begin{split}
  \theta_{j'}\sum_{p} (-1)^{\gr{j}\gr{p}}\theta_{p}\, L_{pj}(x)L_{p'\ell}(y)
  &=\left(\begin{BMAT}[5pt]{c:c}{c:c}
      -(x-y+\kappa)\JD_{n+m}\km & (x+\kappa)\JD_{n+m} \\
      y\JD_{n+m} & 0
    \end{BMAT} \right)_{j,\ell}
\end{split}
\end{equation}
and
\begin{equation}\label{eq:quadsimp2}
  \theta_{i}\sum_{p} (-1)^{\gr{p}+\gr{i}\gr{p}}\theta_{p}\, L_{kp'}(y)L_{ip}(x) =
  \left(\begin{BMAT}[5pt]{c:c}{c:c}
    (x-y+\kappa)\kp\JD_{n+m} & y\JD_{n+m} \\
    (x+\kappa)\JD_{n+m} & 0
  \end{BMAT} \right)_{k,i} \,.
\end{equation}
\end{Lem}

\begin{proof}
Let us verify \eqref{eq:quadsimp1}. For the lower right block, we use that
\begin{equation}
  (-1)^{\gr{i}\gr{j}}\theta_{i'} L_{i'j}(x) + (-1)^\gr{j}\theta_{j}L_{j'i}(x)=0
  \qquad \text{for} \quad n+m+1\leq i,j\leq 2n+2m \,,
\end{equation}
which can be shown using the equality
\begin{equation}\label{eq:ksym}
  \kp_{j,n+m-i+1} = -(-1)^{\gr{i}\gr{j}} \kp_{i,n+m-j+1}
  \qquad \text{for} \quad 1\leq i,j\leq n+m \,.
\end{equation}
For the upper right block, we use that
\begin{equation}
  (-1)^{\gr{j}\gr{\ell}+\gr{j}}\theta_{j}\theta_{\ell}\, L_{\ell j}(x) +
  \sum_{p=1}^{n+m} \theta_{j} (-1)^{\gr{j}\gr{p}+\gr{j}+\gr{p}}\theta_{p}\, L_{p'j}(x)L_{p\ell'}(y)
  = (x+\kappa)\delta_{j\ell}
\end{equation}
for any $1\leq j,\ell\leq n+m$. The proof of this follows from the equality
\begin{equation}\label{eq:L12relation}
  -\sum_{p=1}^{n+m} (-1)^{\gr{j}\gr{p}+\gr{j}+\gr{p}}\, \km_{n+m-p+1,j}\kp_{p,n+m-\ell+1} -
  (-1)^{\gr{j}\gr{\ell}+\gr{\ell}} (\kp\km)_{\ell j} = \kappa\,\delta_{j\ell}
\end{equation}
for any $1\leq j,\ell\leq n+m$. For the lower left block, we use the relation
\begin{equation}
  L_{j\ell}(y) + \theta_{j}\sum_{p=1}^{n+m} (-1)^{\gr{j}\gr{p}}\theta_{p}\, L_{pj'}(x)L_{p'\ell}(y)=y\delta_{j\ell}
  \qquad \text{for} \quad 1\leq j,\ell\leq n+m
\end{equation}
based on
\begin{equation}
  (\kp\km)_{j\ell} + \sum_{p=1}^{n+m}(-1)^{\gr{j}\gr{p}}\, \kp_{p,n+m-j+1}\km_{n+m-p+1,\ell}=0
  \qquad \text{for} \quad 1\leq j,\ell\leq n+m \,,
\end{equation}
which follows from \eqref{eq:ksym}. Finally, for the upper left block, we note that
\begin{equation}
  \sum_{p=1}^{n+m} (-1)^{\gr{j}\gr{p}+\gr{j}} L_{pj}(x)L_{p'\ell}(y) +
  \sum_{p=1}^{n+m} (-1)^{\gr{j}\gr{p}+\gr{j}+\gr{p}} L_{p'j}(x)L_{p\ell}(y) =
  -(x-y+\kappa)\km_{n+m-j+1,\ell}
\end{equation}
for any $1\leq j,\ell\leq n+m$. This relation is equivalent to
\begin{multline}\label{eq:L11relation}
  \sum_{p=1}^{n+m} (-1)^{\gr{j}\gr{p}+\gr{j}} (\kp\km)_{pj}\km_{n+m-p+1,\ell} +
  \sum_{p=1}^{n+m} (-1)^{\gr{j}\gr{p}+\gr{j}+\gr{p}}\, \km_{n+m-p+1,j}(\kp\km)_{p\ell} = \\
  -\kappa\, \km_{n+m-j+1,\ell} \,.
\end{multline}
Due to the symmetry
\begin{equation}
  (-1)^{\gr{j}\gr{\ell}+\gr{j}+\gr{\ell}}\, \km_{n+m-\ell+1,j} = -\km_{n+m-j+1,\ell}
  \qquad \text{for} \quad 1\leq j,\ell\leq n+m \,,
\end{equation}
the relation \eqref{eq:L11relation} can be easily verified using \eqref{eq:L12relation}.

The proof of~\eqref{eq:quadsimp2} is completely analogous; we leave details to the interested reader.
\end{proof}

\begin{Rem}\label{rem:BCD-comparison-1}
In particular, we recover orthogonal and symplectic degenerate linear Lax~matrices:
\begin{itemize}

\item
For $m=0$, we recover precisely the $D_n$-type Lax matrix of~\cite[(2.231)]{ft1}, constructed first in~\cite[(4.3)]{f}.

\item
For $n=0$, we get
\begin{equation}
  L(-x)=
  \left(\begin{BMAT}[5pt]{c:c}{c:c}
    -x\ID_{m}+\bp\bm & \bp \\
    \bm & \ID_{m}
  \end{BMAT} \right) =
  -L_{FKT}(x)
  \left(\begin{BMAT}[5pt]{c:c}{c:c}
    \ID_{m} & 0 \\
    0 & -\ID_{m}
  \end{BMAT} \right) \,,
\end{equation}
where $L_{FKT}(x)$ is the $C_m$-type Lax matrix of~\cite[(8.55)]{fkt}, discovered first in~\cite[(3.50)]{ft1}.
We note that the constant matrix above is $\ID_\theta$ of~\eqref{eq:ID-theta matrix} and hence it can be dropped out
of the RTT-relation due to the invariance $[R(x),\ID_\theta \otimes \ID_\theta]=0$ established in Lemma~\ref{lem:jt}.

\end{itemize}
\end{Rem}


\subsection{Non-degenerate linear Lax matrix through the fusion of two degenerate}\label{sec:Facllm}
\

Similarly to the $\gl(n|m)$-case of Section~\ref{sec:Atype}, one can fuse two degenerate orthosymplectic Lax
matrices to obtain a non-degenerate linear one. The main result of this subsection is Proposition~\ref{prop:lin}.

Evoking the Lax matrix $L(x)$ of~\eqref{eq:linlaxdeg}, we define
\begin{equation}\label{eq:conjL}
  \bL_\theta(x)=\JD_\theta L(x) \JD_\theta^{-1}
\end{equation}
with $\JD_\theta$ as in~\eqref{eq:jd-theta}:
\begin{equation}
  \JD_\theta=
  \left(\begin{BMAT}[5pt]{c:c}{c:c}
    0 & -\JD_{n} \\
    \JD_{n+2m } & 0
  \end{BMAT} \right)
    \,, \qquad
  \JD_\theta^{-1}=
  \left(\begin{BMAT}[5pt]{c:c}{c:c}
    0 & \JD_{n+2m} \\
    -\JD_{n} & 0
  \end{BMAT} \right) \,.
\end{equation}
Due to the invariance $[R(x),\JD_\theta\otimes \JD_\theta]=0$ of the orthosymplectic $R$-matrix~\eqref{eq:R-osp},
established in Lemma~\ref{lem:jjt}, the matrix $\bL_\theta(x)$ in \eqref{eq:conjL} is a solution to the same
RTT-relation \eqref{eq:rttA}, hence, is a Lax matrix.

\begin{Rem}
In contrast to~\eqref{eq:hat-L}, $\BZ_2$-gradings of the Lax matrices $\bL_\theta(x)$ and $L(x)$ coincide.
\end{Rem}

Explicitly, we have:
\begin{equation}\label{eq:barmatrix}
  \bL_\theta(y)=\,
  \left(\begin{BMAT}[5pt]{c:c}{c:c}
    \ID_{n+m} & \km^\theta \\
    \kp^\theta & y\ID_{n+m} + \kp^\theta\km^\theta
  \end{BMAT}\right)
\end{equation}
with
\begin{equation}\label{eq:linear-Lax-D}
  \kp^\theta=
  \left(\begin{BMAT}[5pt]{c:c}{c:c}
    \bar{\mathbf{C}}^t & \JD_m \bp \JD_m \\
    -\JD_n \ap \JD_n & \JD_n \bar{\mathbf{C}} \JD_m
  \end{BMAT} \right)
    \qquad \mathrm{and} \qquad
  \km^\theta=
  \left(\begin{BMAT}[5pt]{c:c}{c:c}
    {\mathbf{C}}^t & \JD_n \am \JD_n \\
    \JD_m \bm \JD_m & -\JD_m {\mathbf{C}} \JD_n
  \end{BMAT} \right) \,.
\end{equation}
We further apply the fermionic particle-hole transformation
\begin{equation}\label{eq:phc}
  \oc_{ij} \mapsto  \ocd_{ji} \,,\quad \ocd_{ji} \mapsto  \oc_{ij}
  \qquad \mathrm{for} \quad 1\leq j\leq n \,,\ n+m+1\leq i\leq n+2m \,,
\end{equation}
as well as the bosonic particle-hole transformation
\begin{equation}\label{eq:pha}
  \oa_{ij} \mapsto -\oad_{ji} \,,\quad \oad_{ji} \mapsto \oa_{ij}
  \quad \mathrm{for} \quad 1\leq j\leq n-1 \,,\ n+2m+1\leq i\leq 2n+2m-j \,,
\end{equation}
\begin{equation}\label{eq:phb}
  \ob_{ij} \mapsto (1+\delta_{ij'}) \obd_{ji} \,,\quad \obd_{ji} \mapsto -\frac{1}{1+\delta_{ij'}} \ob_{ij}
  \quad \mathrm{for} \quad n+1\leq j\leq n+m < i\leq 2n+2m+1-j \,.
\end{equation}
This yields the following Lax matrix:
\begin{equation}\label{eq:ospLax-degen-2}
  \bL(y) = \bL_\theta(y)|_{p.h.} =
  \left(\begin{BMAT}[5pt]{c:c}{c:c}
    \ID_{n+m} & \kp \\
    \km & y\ID_{n+m} + \km\kp
  \end{BMAT}\right)
\end{equation}
with $\kp,\km$ precisely as in~\eqref{eq:Koscillators}.

\begin{Rem}\label{rem:BCD-comparison-2}
In particular, we recover orthogonal and symplectic degenerate linear Lax~matrices:
\begin{itemize}

\item
For $m=0$, we recover precisely the $D_n$-type Lax matrix $L_{(-,\ldots,-)}(x)$ of~\cite[(8.84)]{fkt}
(after swapping indices in the generators, i.e.\ $\oa_{ij}\mapsto \oa_{ji}$, $\oad_{ji}\mapsto \oad_{ij}$).

\item
For $n=0$, we get
\begin{equation}
  \bL(-x)=
  \left(\begin{BMAT}[5pt]{c:c}{c:c}
    \ID_m & \bp \\
    -\bm & -x\ID_{m} - \bm\bp
  \end{BMAT} \right) =
  \left(\begin{BMAT}[5pt]{c:c}{c:c}
    \ID_{m} & 0 \\
    0 & -\ID_{m}
  \end{BMAT} \right) L^-_{FKT}(x) \,,
\end{equation}
where $L^-_{FKT}(x)$ is the $C_m$-type Lax matrix equivalent to $L_{(-,\ldots,-)}(x)$ of~\cite[(8.56)]{fkt}
  (after swapping indices and redistributing factor 2, i.e.\
   $\ob_{ij}\mapsto (1+\delta_{ij'})\ob_{ji}$, $\obd_{ji}\mapsto \frac{1}{1+\delta_{ij'}}\obd_{ij}$).
We note that the constant matrix above is $\ID_\theta$ of~\eqref{eq:ID-theta matrix} and hence it can be dropped out
of the RTT-relation due to the invariance $[R(x),\ID_\theta \otimes \ID_\theta]=0$ established in Lemma~\ref{lem:jt}.

\end{itemize}
\end{Rem}

Let us now consider two copies of mutually supercommuting superoscillators labelled by the extra superscript $i=1,2$,
which will be now encoded by the corresponding matrices $\ap_i,\am_i,\bp_i,\bm_i,\bar{\mathbf{C}}_i,\mathbf{C}_i$,
hence also, the subscript in $\kp_i,\km_i$. We consider the following two orthosymplectic Lax matrices:
\begin{equation}\label{eq:twomatrices}
  L^{[1]}(x) = \,
  \left(\begin{BMAT}[5pt]{c:c}{c:c}
      x\ID_{n+m}-\kp_1\km_1 & \kp_1 \\
      -\km_1 & \ID_{n+m}
  \end{BMAT}\right)
    \,, \qquad
  \bL^{[2]}(y)=\,
  \left(\begin{BMAT}[5pt]{c:c}{c:c}
    \ID_{n+m} & \kp_2 \\
    \km_2 & y\ID_{n+m} + \km_2\kp_2
  \end{BMAT}\right)
\end{equation}
of~(\ref{eq:linlaxdeg},~\ref{eq:ospLax-degen-2}) with
\begin{equation}\label{eq:linear-Lax-D2}
  \kp_i=
  \left(\begin{BMAT}[5pt]{c:c}{c:c}
    \bar{\mathbf{C}}_i & \ap_i \\
    \bp_i & -\JD_m \bar{\mathbf{C}}_i^t \JD_n
  \end{BMAT} \right)
    \qquad \mathrm{and} \qquad
  \km_i=
  \left(\begin{BMAT}[5pt]{c:c}{c:c}
    {\mathbf{C}}_i & -\bm_i \\
    \am_i & \JD_n {\mathbf{C}}^t_i \JD_m
  \end{BMAT} \right) \,.
\end{equation}

The Lax matrices in \eqref{eq:twomatrices} obey the following factorisation formula,
cf.~(\ref{eq:linear factorization1},~\ref{eq:tra121}):
\begin{equation}\label{eq:linear factorization}
  L^{[1]}(x)\bL^{[2]}(y) = \,
  \left(\begin{BMAT}[5pt]{c:c}{c:c}
     x\ID_{n+m}-\kp_1'\km_1' & \left((y-x )\ID_{n+m}+\kp_1'\km_1'\right)\kp_1' \\
     -\km_1' & y\ID_{n+m}+\km_1'\kp_1'
  \end{BMAT}\right)
    \, . \,
  \left(\begin{BMAT}[5pt]{c:c}{c:c}
    \ID_{n+m} & \kp'_2 \\
    0 & \ID_{n+m}
  \end{BMAT}\right)
\end{equation}
where we set
\begin{equation}\label{eq:tra12}
\begin{split}
  & \km'_1=\km_1-\km_2 \,, \qquad \kp'_1=\kp_1 \,, \\
  & \kp'_2=\kp_2+\kp_1 \,, \qquad \km'_2=\km_2 \,.
\end{split}
\end{equation}
Moreover, similarly to~(\ref{eq:Atype-similarity},~\ref{eq:S-Atype}), the transformation \eqref{eq:tra12}
can be expressed through a similarity transformation in the superoscillator space:
\begin{equation}
  \km_i' = \mathbf{S} \km_i \mathbf{S}^{-1} \,,\qquad
  \kp_i'= \mathbf{S}\kp_i \mathbf{S}^{-1} \qquad (i=1,2)
\end{equation}
with
\begin{equation}\label{eq:Ksim}
  \mathbf{S} =
  \exp \left[
    \sum_{i,j}\oad_{ji}^{[1]}\oa_{ij}^{[2]} + \sum_{i,j}\obd_{ji}^{[1]}\ob_{ij}^{[2]} + \sum_{i,j} \ocd_{ji}^{[1]}\oc_{ij}^{[2]}
  \right] \,,
\end{equation}
where the indices $i,j$ take all possible values as given in~\eqref{eq:phc}--\eqref{eq:phb}.
We note that all the summands in the exponent above are bosonic and pairwise supercommute.

We thus obtain the key result of this subsection:

\begin{Prop}\label{prop:lin}
(a) For any $x_1,x_2\in\mathbb{C}$, the matrix
\begin{equation}\label{eq:linlaxnorm}
  \mathcal{L}_{x_1,x_2}(x)=
  \left(\begin{BMAT}[5pt]{c:c}{c:c}
    (x+x_1)\ID_{n+m}-\kp_1\km_1 & \left((x_2-x_1)\ID_{n+m}+\kp_1\km_1\right)\kp_1 \\
     -\km_1 & (x+x_2)\ID_{n+m}+\km_1\kp_1
  \end{BMAT}\right)
\end{equation}
is a solution to the RTT-relation \eqref{eq:rttA} with the $R$-matrix~\eqref{eq:R-osp}, hence, is a Lax matrix.

\medskip
\noindent
(b) The Lax matrix from (a) is a fusion of two degenerate Lax matrices in \eqref{eq:twomatrices} through:
\begin{equation}\label{eq:linear factorizationLax1}
  L^{[1]}(x+x_1)\bL^{[2]}(x+x_2)=
  \mathbf{S} \, \mathcal{L}_{x_1,x_2}(x) \,
  \left(\begin{BMAT}[5pt]{c:c}{c:c}
    \ID_{n+m} & \kp_2\\
    0 & \ID_{n+m}
  \end{BMAT}\right)
  \mathbf{S}^{-1}
\end{equation}
with the similarity transformation $\mathbf{S}$ of \eqref{eq:Ksim}.
\end{Prop}

\begin{Rem}\label{rem:BCD-comparison-3}
Similarly to Remarks~\ref{rem:BCD-comparison-1} and~\ref{rem:BCD-comparison-2}, we note that the orthosymplectic Lax matrix
$\mathcal{L}_{x_1,x_2}(x)$ of~\eqref{eq:linlaxnorm} generalizes the authors' previous work for orthogonal and symplectic types:
\begin{itemize}

\item
For $m=0$, we recover precisely the $D_n$-type Lax matrix of~\cite[(6.2)]{fkt} when setting $x_1=t$ and $x_2=-t-n+1$.

\item
For $n=0$, setting $x_1=-t$ and $x_2=t+m+1$, we recover the $C_m$-type Lax matrix $\mathcal{L}_{FKT}(-x)$ of~\cite[(5.3)]{fkt}:
\begin{equation}
\begin{split}
  \mathcal{L}_{x_1=-t,x_2=t+m+1}(x) =
  \left(\begin{BMAT}[5pt]{c:c}{c:c}
    (x-t)\ID_{m} + \bp\bm & \bp\left((2t+m+1)\ID_{m}-\bm\bp\right) \\
    \bm & (x+t+m+1)\ID_{m}-\bm\bp
  \end{BMAT}\right) =
  -\mathcal{L}_{FKT}(-x) \,.
\end{split}
\end{equation}

\end{itemize}
\end{Rem}

\begin{Rem}\label{rem:normalized-limit-osp}
Similarly to Remark~\ref{rem:Atype-normalized-limits} and generalizing~\cite[\S8]{fkt}, we can vice versa obtain
the matrices $L(x)$ and $\bL(x)$ of~(\ref{eq:linlaxdeg},~\ref{eq:ospLax-degen-2}) from the non-degenerate linear
Lax matrix $\mathcal{L}_{x_1,x_2}(x)$ of \eqref{eq:linlaxnorm} via the \emph{renormalized limit} procedures
(which clearly preserve the property of being Lax):
\begin{equation}\label{eq:osp renormalized limit}
\begin{split}
  & L(x) =
    \lim_{t\to \infty}\ \Big\{\mathcal{L}_{0,t}(x)\cdot \mathrm{diag}
    \Big(\underbrace{1,\ldots,1}_{n+m};\underbrace{\sfrac{1}{t},\ldots,\sfrac{1}{t}}_{n+m}\Big)\Big\} \,, \\
  & \bL(x) = \lim_{t\to \infty}\ \Big\{\mathrm{diag}
    \Big(\underbrace{\sfrac{1}{t},\ldots,\sfrac{1}{t}}_{n+m};\underbrace{1,\ldots,1}_{n+m}\Big)
    \cdot \mathcal{L}_{t,0}(x)\Big\}\Big|_{\kp \mapsto -\kp \,,\, \km \mapsto -\km}  \,.
\end{split}
\end{equation}
\end{Rem}

\begin{Rem}\label{rem:linear-Lax-all-parities}
According to Remark~\ref{rem:different parities osp}, the Lax matrix $L(x)$ from Theorem~\ref{thm:degenerate-linear-Lax}
and the Lax matrix $\mathcal{L}_{x_1,x_2}(x)$ from Proposition~\ref{prop:lin}(a) give rise to the corresponding degenerate
and non-degenerate linear Lax matrices for all $\BZ_2$-gradings of $V$.
\end{Rem}


\section{Quadratic orthosymplectic Lax matrices}\label{sec:quadratic-Lax}

In this section, we investigate some quadratic superoscillator orthosymplectic Lax matrices.


\subsection{From linear to quadratic Lax matrices}\label{sec:fl2q}
\

In this subsection, we consider a factorisation formula different from the one presented in Subsection~\ref{sec:Facllm}
that allows to derive a degenerate quadratic Lax matrix, see Theorem~\ref{prop:quadL}.

Consider the Lax matrix of \eqref{eq:linlaxdeg} written in the block form with blocks on the diagonal
of size $1\times 1$, $(n+m-1)\times (n+m-1)$, $(n+m-1)\times (n+m-1)$, and $1\times 1$, respectively:
\begin{equation}\label{eq:lindec}
  L(x)=
  \left(\begin{BMAT}[5pt]{c|c:c|c}{c|c:c|c}
    x-\vp\vm & -\vp\kmm & \vp & 0 \\
    -\kpp\vm & x\id_{n+m-1} - \kpp\kmm + \idb_{n+m-1}\vp^t\vm^t\idg^{-1}_{n-1,m} & \kpp & -\idb_{n+m-1}\vp^t \\
    -\vm & -\kmm & \id_{n+m-1} & 0 \\
    0 & -\vm^t\idg^{-1}_{n-1,m} & 0 & 1
  \end{BMAT} \right) \,.
\end{equation}
Here, we expressed $\kp,\km$ of~\eqref{eq:Koscillators} via
\begin{equation}\label{eq:kpkm}
  \kp =
  \left(\begin{BMAT}[5pt]{c|c}{c|c}
    \vp & 0 \\
    \kpp & -\idb_{n+m-1} \vp^t
  \end{BMAT}\right)
    \,,\qquad
  \km =
  \left(\begin{BMAT}[5pt]{c|c}{c|c}
    \vm & \kmm \\
    0 & \vm^t\idg^{-1}_{n-1,m}
  \end{BMAT}\right)
\end{equation}
with
\begin{equation}\label{eq:gmat}
 \idg_{n-1,m} =
 \left(\begin{BMAT}[5pt]{c:c}{c:c}
   0 & -\JD_{n-1} \\
   \JD_m & 0
 \end{BMAT}\right)
\end{equation}
as in~\eqref{eq:hatJd-idg}, and used the $(n+m-1)$-dimensional vectors
\begin{equation}\label{eq:vp}
  \vp=
  \left(\begin{array}{cccccc}
    \ocd_{1,n+m+1} & \cdots & \ocd_{1,n+2m} & \oad_{1,n+2m+1} & \cdots & \oad_{1,2n+2m-1}
  \end{array} \right)\,,
\end{equation}
\begin{equation}\label{eq:vm}
  \vm=
  \left(\begin{array}{cccccc}
    \oc_{n+m+1,1} & \ldots & \oc_{n+2m,1} & \oa_{n+2m+1,1} & \cdots & \oa_{2n+2m-1,1}
  \end{array}\right)^t \,.
\end{equation}
Finally, the submatrices $\km^\circ,\kp^\circ$ of $\km,\kp$ appearing in \eqref{eq:kpkm} read as follows:
\begin{equation}\label{eq:Koscillators2}
  \kp^\circ=
  \left(\begin{BMAT}[5pt]{c:c}{c:c}
    \bar{\mathbf{C}}^\circ & \ap^\circ \\
    \bp & -\JD_m (\bar{\mathbf{C}}^\circ)^t \JD_{n-1}
  \end{BMAT} \right)
    \qquad \mathrm{and} \qquad
  \km^\circ=
  \left(\begin{BMAT}[5pt]{c:c}{c:c}
    {\mathbf{C}}^\circ & -\bm \\
    \am^\circ & \JD_{n-1} ({\mathbf{C}}^\circ)^t \JD_m
  \end{BMAT} \right) \,,
\end{equation}
where the block matrices $\ap^\circ,\am^\circ$ and $\bar{\mathbf{C}}^\circ,{\mathbf{C}}^\circ$
encode the bosonic and fermionic superoscillators via
\begin{equation}\label{eq:ApAmD}
  \ap^\circ=
  \left(\begin{array}{cccc}
    \oad_{2,n+2m+1} & \cdots & \oad_{2,2n+2m-2} & 0 \\
    \vdots & \iddots & 0 & - \oad_{2,2n+2m-2} \\
    \oad_{n-1,n+2m+1} & 0 & \iddots & \vdots \\
    0 & -\oad_{n-1,n+2m+1} & \cdots & -\oad_{2,n+2m+1}
  \end{array}\right) \,,
\end{equation}
\begin{equation}
  \am^\circ=
  \left(\begin{array}{cccc}
    \oa_{n+2m+1,2} & \cdots & \oa_{n+2m+1,n-1} & 0 \\
    \vdots & \iddots & 0 & -\oa_{n+2m+1,n-1} \\
    \oa_{2n+2m-2,2} & 0 & \iddots & \vdots \\
    0 & -\oa_{2n+2m-2,2} & \cdots & - \oa_{n+2m+1,2}
  \end{array}\right)\,,
\end{equation}
\begin{equation}\label{eq:ApAmSUSY2}
  \bar{\mathbf{C}}^\circ=
  \left(\begin{array}{ccc}
    \ocd_{2,n+m+1} & \cdots & \ocd_{2,n+2m}  \\
    \vdots & \ddots & \vdots \\
    \ocd_{n,n+m+1} & \cdots & \ocd_{n,n+2m}
  \end{array}\right)
    \,, \qquad
  {\mathbf{C}}^\circ=
  \left(\begin{array}{ccc}
    \oc_{n+m+1,2} & \cdots & \oc_{n+m+1,n} \\
    \vdots & \ddots & \vdots \\
    \oc_{n+2m,2} & \cdots & \oc_{n+2m,n}
  \end{array}\right) \,,
\end{equation}
while the block matrices $\bp,\bm$ encoding the bosonic superoscillators
are precisely as in (\ref{eq:ApAmCp},~\ref{eq:ApAmCm}).

To construct the degenerate quadratic Lax matrix, we define
\begin{equation}\label{eq:conjL2}
  \hat L_\theta(x) = \hat \JD_\theta L(x) \hat\JD_\theta^{-1}
\end{equation}
with $\hat \JD_\theta$ as in~\eqref{eq:hatJd-idg}:
\begin{equation}
  \hat \JD_\theta=
  \left(\begin{BMAT}[5pt]{c|c:c|c}{c|c:c|c}
    1 & 0 & 0 & 0 \\
    0 & 0 & \idg_{n-1,m} & 0 \\
    0 & \JD_{n+m-1} & 0 & 0 \\
    0 & 0 & 0 & -1
  \end{BMAT} \right)
    \,,\qquad
  \hat \JD_\theta^{-1}=
  \left(\begin{BMAT}[5pt]{c|c:c|c}{c|c:c|c}
    1 & 0 & 0 & 0 \\
    0 & 0 & \JD_{n+m-1} & 0 \\
    0 & \idg_{n-1,m}^t & 0 & 0 \\
    0 & 0 & 0 & -1
  \end{BMAT} \right)
\end{equation}
as $\idg^t_{n-1,m}=\idg^{-1}_{n-1,m}$. The matrix $\hat L_\theta(x)$ in \eqref{eq:conjL2} is again a solution to the
RTT-relation~\eqref{eq:rttA} with the $R$-matrix~\eqref{eq:R-osp}, hence, is a Lax matrix. This follows from the invariance
$[R(x),\hat \JD_\theta \otimes \hat \JD_\theta]=0$ of the orthosymplectic $R$-matrix~\eqref{eq:R-osp}, established in Lemma~\ref{lem:hjt}.

We rename the fermionic superoscillators
\begin{equation}\label{eq:rename-1}
  \oc_{i1}\mapsto \oc_{i'1} \,,\quad  \ocd_{1i}\mapsto \ocd_{1i'} \qquad \mathrm{for} \quad n+m+1\leq i\leq n+2m
\end{equation}
and also rename the bosonic superoscillators
\begin{equation}
  \oa_{i1} \mapsto -\oa_{i'1} \,,\qquad  \oad_{1i}\mapsto -\oad_{1i'} \qquad \mathrm{for} \quad  n+2m+1\leq i\leq 2n+2m-1 \,.
\end{equation}
We further apply the following particle-hole transformation of the remaining superoscillators:
\begin{equation}
  \oc_{ij} \mapsto \ocd_{ji} \,,\qquad  \ocd_{ji}\to \oc_{ij}
  \qquad \mathrm{for} \quad  2\leq j\leq n \,,\ n+m+1\leq i\leq n+2m \,,
\end{equation}
\begin{equation}
  \oa_{ij}\mapsto -\oad_{ji} \,,\qquad   \oad_{ji}\to \oa_{ij}
  \qquad \mathrm{for} \quad  2\leq j\leq n-1 \,,\ n+2m+1\leq i\leq 2n+2m-j \,,
\end{equation}
\begin{equation}\label{eq:rename-5}
  \ob_{ij} \mapsto (1+\delta_{ij'}) \obd_{ji} \,,\qquad  \obd_{ji} \mapsto\frac{-1}{1+\delta_{ij'}} \ob_{ij}
  \quad \mathrm{for} \quad  n+1\leq j\leq n+m<i\leq 2n+2m+1-j \,.
\end{equation}
Applying the above operations~\eqref{eq:rename-1}--\eqref{eq:rename-5} to~\eqref{eq:conjL2} yields the following Lax matrix:
\begin{equation}\label{eq:LaxZwei}
  \hat L(y) = \hat L_\theta(y)|_{p.h.} =
  \left(\begin{BMAT}[5pt]{c|c:c|c}{c|c:c|c}
    y-\wp\wm & \wp & \wp\kpp & 0 \\
    -\wm & \id & \kpp & 0 \\
    -\kmm\wm & \kmm & y\id+\kmm\kpp + \idg^{-1}\wp^t\wm^t\idb & \idg^{-1} \wp^t \\
    0 & 0 & \wm^t\idb & 1
  \end{BMAT} \right)
\end{equation}
with
\begin{equation}\label{eq:wp}
  \wp=
  \left(\begin{array}{cccccc}
      \oad_{1,2} & \cdots & \oad_{1,n} & \ocd_{1,n+1} & \cdots & \ocd_{1,n+m}
  \end{array} \right) \,,
\end{equation}
\begin{equation}\label{eq:wm}
  \wm=
  \left(\begin{array}{cccccc}
    \oa_{2,1} & \ldots & \oa_{n,1} & \oc_{n+1,1} & \cdots & \oc_{n+m,1}
  \end{array}\right)^t \,,
\end{equation}
and the $(n+m-1)\times(n+m-1)$ matrices $\id$, $\idb$, $\idg$ defined via:
\begin{equation}\label{eq:no-indices}
  \id=\id_{n+m-1} \,,\qquad \idb=\idb_{n+m-1} \,,\qquad \idg=\idg_{n-1,m} \,.
\end{equation}

Similarly to~\eqref{eq:linear factorization1} and~\eqref{eq:linear factorization}, let us factorise the product of
the Lax matrices in \eqref{eq:lindec} and \eqref{eq:LaxZwei}. As before, we shall use the subscript $i=1,2$ to
distinguish between the superoscillators and the corresponding matrices (using $i=1$ in the context of $L(x)$,
$i=2$ in the context of $\hat L(y)$). One has the following factorisation:
\begin{equation}\label{eq:factorisatio-no3}
  L^{[1]}(x) \hat L^{[2]}(y) =
  \mathbf{L}'(x,y)
  \left(\begin{BMAT}[5pt]{c|c:c|c}{c|c:c|c}
    1 & 0 & 0 & 0 \\
    0 & \id & (\kp_{2}^\circ)' & 0 \\
    0 & 0 & \id & 0 \\
    0 & 0 & 0 & 1
  \end{BMAT}\right)
\end{equation}
where
\begin{equation}
  (\kp_{2}^\circ)'=\kp_{2}^\circ + \kp_{1}^\circ
  \,,\qquad (\km_{2}^\circ)'=\km_{2}^\circ \,,
\end{equation}
and $\mathbf{L}'(x,y)$ is further factorised as
\begin{equation}\label{eq:quadl}
{\small
  \mathbf{L}'(x,y)=
  \left(\begin{BMAT}[5pt]{c|c|c}{c|c|c}
    1 & \up' & \frac{1}{2}\up'M(\up')^t \\
    0 & \id_{N+2m-2} & M(\up')^t \\
    0 & 0 & 1
  \end{BMAT}\right)
  \left(\begin{BMAT}[5pt]{c|c|c}{c|c|c}
    x\cdot y & 0 & 0 \\
    0 & \mathcal{L}'(x,y) & 0 \\
    0 & 0 & 1
  \end{BMAT}\right)
  \left(\begin{BMAT}[5pt]{c|c|c}{c|c|c}
    1 & 0 & 0 \\
    -\um' & \id_{N+2m-2} & 0 \\
    \frac{1}{2}(\um')^t M\um' & -(\um')^t M & 1
  \end{BMAT}\right)
}
\end{equation}
with
\begin{equation}\label{eq:Mmatrix}
  M=
  \left(\begin{BMAT}[5pt]{c:c}{c:c}
    0 & -\idb \\
    \idg^{-1} & 0
  \end{BMAT}\right)  \,.
\end{equation}
Let us now describe $\up', \um'$, and $\mathcal{L}'(x,y)$ featuring in~\eqref{eq:quadl}.
The first two are defined via
\begin{equation}\label{eq:uprime}
  \up'=
  \left(\begin{array}{cc}
    \wp' & \vp'
  \end{array}\right)
    \,, \qquad
  \um'=\left(\begin{array}{c}
    \wm'\\
    \vm'
  \end{array}\right) \,,
\end{equation}
where the vectors $\vp',\vm'$ and $\wm',\wp'$ are given by (cf.~(\ref{eq:vp}, \ref{eq:vm}, \ref{eq:Koscillators2}, \ref{eq:wp}, \ref{eq:wm}))
\begin{equation}\label{eq:transvw1}
  \vp'=\vp -\wp\kpp_1 \,,\qquad \vm'=\vm \,,
\end{equation}
\begin{equation}\label{eq:transvw2}
  \wm'=\wm+\kpp_1\vm \,,\qquad \wp'=\wp \,.
\end{equation}
The matrix $\mathcal{L}'(x,y)$ in the middle factor of the factorisation \eqref{eq:quadl} reads (cf.~\eqref{eq:linlaxnorm})
\begin{equation}\label{eq:innerLax}
  \mathcal{L}'(x,y)=
  \left(\begin{BMAT}[5pt]{c:c}{c:c}
    x\ID-(\kpp_1)'(\kmm_1)' & \left((y-x)\ID+(\kpp_1)'(\kmm_1)'\right)(\kpp_1)' \\
    - (\kmm_1)' & y\ID_{}+(\kmm_1)'(\kpp_1)'
  \end{BMAT}\right) \,
\end{equation}
where
\begin{equation}\label{eq:transKm}
  (\kmm_1)'=\kmm_1-\kmm_2+\vm\wp-\idg^{-1}\wp^t\vm^t\idg^{-1}\,,\qquad (\kpp_1)'=\kpp_1  \,.
\end{equation}

\begin{Rem}
In the derivation of \eqref{eq:quadl}, we used the following two equalities:
\begin{equation}\label{eq:transpose-consistency-1}
  (\vp')^t=\vp^t - \idb\kpp_1\idg^{-1}\wp^t \,,\qquad
  (\wm')^t=\wm^t+\vm^t\idg^{-1} \kpp_1\idb \,,
\end{equation}
where $\vp'$ and $\wm'$ are given by~\eqref{eq:transvw1} and~\eqref{eq:transvw2}, respectively.
To verify \eqref{eq:transpose-consistency-1}, we use
\begin{equation}
  \wp\idb\vp^t = -\vp\idg^{-1}\wp^t \,,\qquad
  \wm^t\idb \vm=-\vm^t\idg^{-1}\wm \,,
\end{equation}
as well as
\begin{equation}\label{eq:vvww}
  \vm^t\idg^{-1}\kpp_1\vm = 0 = \wp\kpp_1\idg^{-1}\wp^t
    \qquad \mathrm{and} \qquad
  \wp\idb(\wp\kpp_1)^t=0=(\kpp_1\vm)^t\idb\vm  \,,
\end{equation}
where we take into account the $\BZ_2$-grading of the superoscillators.
It also follows from \eqref{eq:vvww} that
\begin{equation}
  \up' M(\up')^t = \up M(\up)^t \,,\qquad (\um')^t M\um'=(\um)^t M\um
\end{equation}
with
\begin{equation}
  \up=
  \left(\begin{array}{cc}
    \wp & \vp
  \end{array}\right)
    \,,\qquad
  \um=
  \left(\begin{array}{c}
    \wm\\
    \vm
  \end{array}\right) \,.
\end{equation}
Let us also note that the first and third terms in the right-hand side of the factorisation \eqref{eq:quadl}
can be written as exponential expansions that truncate after the second term, cf.~\cite[(3.2)]{f}.
\end{Rem}

Finally, we note that there is a similarity transformation in the superoscillator space such that
\begin{equation}\label{eq:transK}
\begin{split}
  & (\kmm_1)'=\mathbf{S}\kmm_1\mathbf{S}^{-1} \,,\quad
    (\kpp_2)'=\mathbf{S}\kpp_2\mathbf{S}^{-1} \,,\quad
    (\kpp_1)'=\mathbf{S}\kpp_1\mathbf{S}^{-1} \,,\quad
    (\kmm_2)'=\mathbf{S}\kmm_2\mathbf{S}^{-1} \,,  \\
  & \qquad \quad
    \wm'=\mathbf{S}\wm \mathbf{S}^{-1} \,,\qquad
    \vp'=\mathbf{S}\vp \mathbf{S}^{-1} \,,\qquad
    \wp'=\mathbf{S}\wp \mathbf{S}^{-1} \,,\qquad
    \vm'=\mathbf{S}\vm \mathbf{S}^{-1} \,.
\end{split}
\end{equation}
Here, $\mathbf{S}$ is given explicitly by
\begin{equation}\label{eq:linear+cubic-S}
  \mathbf{S}=\mathbf{S}_\circ \mathbf{S}_3
\end{equation}
with
\begin{equation}\label{eq:cubic-S}
  \mathbf{S}_3 = \exp \left(-\wp \kpp_1\vm \right)
\end{equation}
and
\begin{equation}\label{eq:Ksim2}
  \mathbf{S}_\circ =
  \exp\left[
    \sum_{i,j}\oad_{ji}^{[1]}\oa_{ij}^{[2]} + \sum_{i,j}\obd_{ji}^{[1]}\ob_{ij}^{[2]} + \sum_{i,j} \ocd_{ji}^{[1]}\oc_{ij}^{[2]}
  \right] \,,
\end{equation}
where the indices $i,j$ take all possible values as they appear in~\eqref{eq:Koscillators2}.
It is obvious that the summands in the exponents of \eqref{eq:cubic-S} and \eqref{eq:Ksim2},
respectively, are bosonic and pairwise commute (however, $\mathbf{S}_\circ$ and $\mathbf{S}_3$ do not commute!).

\begin{Rem}
The equalities in \eqref{eq:transK} are verified by direct computations. Let us only evaluate the transformation
of $\km^\circ_1$ under $\mathbf{S}_3$ (the rest of computations being much simpler), cf.~\eqref{eq:transKm}:
\begin{equation}
\begin{split}
  \mathbf{S}_3( \km^\circ_1)_{ij} \mathbf{S}_3^{-1}
  & = (\km^\circ_1)_{ij}-\wp \kpp_1\vm( \km^\circ_1)_{ij}+( \km^\circ_1)_{ij} \wp \kpp_1\vm \\
  & = (\km^\circ_1)_{ij}-\wp_k( \kpp_1)_{k\ell}\vm_\ell( \km^\circ_1)_{ij} + (\km^\circ_1)_{ij} \wp_k (\kpp_1)_{k\ell}\vm_\ell \\
  & = (\km^\circ_1)_{ij}-\wp_k( \kpp_1)_{k\ell}\vm_\ell( \km^\circ_1)_{ij}+(-1)^{(\gr{n+m-i+1}+\gr{j+1})\gr{k+1}}
      \wp_k(\km^\circ_1)_{ij}  (\kpp_1)_{k\ell}\vm_\ell \\
  & = ( \km^\circ_1)_{ij} +\vm_i\wp_j-(\idg^{-1}\wp^t)_i(\vm^t \idg^{-1})_j \,,
\end{split}
\end{equation}
where we summed over all possible $k,\ell$. The last equality above follows from \eqref{eq:comK} which yields
\begin{equation}
\begin{split}
  & (-1)^{(\gr{n+m-i+1}+\gr{j+1})\gr{k+1}}\wp_k( \km^\circ_1)_{ij} (\kpp_1)_{k\ell}\vm_\ell = \\
  &   \qquad \qquad \wp_k (\kp^\circ_1)_{k\ell} \vm_\ell (\km^\circ_1)_{ij}+\vm_i\wp_j
     - (-1)^{\gr{n+m-i+1}+\gr{j+1}} \wp_{n+m-i} \vm_{n+m-j}
\end{split}
\end{equation}
with
\begin{equation}
  (\idg^{-1}\wp^t)_i = -(-1)^{\gr{n+m-i+1}}\wp_{n+m-i} \,,\qquad  (\vm^t \idg^{-1})_j=-(-1)^\gr{j+1}\vm_{n+m-j} \,.
\end{equation}
\end{Rem}

We thus obtain the key result of this subsection:

\begin{Thm}\label{prop:quadL}
(a) For any $x_1,x_2\in\mathbb{C}$, the matrix
\begin{equation}\label{eq:quadl21}
\begin{split}
  &\mathbf{L}_{x_1,x_2}(x)=\hfill \\[0.3cm]
  &\left(\begin{BMAT}[5pt]{c|c|c}{c|c|c}
     1 & \up & \frac{1}{2}\up M\up^t \\
     0 & \id_{N+2m-2} & M\up^t \\
     0 & 0 & 1
  \end{BMAT}\right)
  \left(\begin{BMAT}[5pt]{c|c|c}{c|c|c}
   (x+x_1)(x+x_2) & 0 &  0 \\
   0 &\mathcal{L}_{x_1,x_2}(x) & 0 \\
   0 & 0 & 1
  \end{BMAT}\right)
  \left(\begin{BMAT}[5pt]{c|c|c}{c|c|c}
    1 & 0 & 0 \\
    -\um & \id_{N+2m-2} & 0 \\
    \frac{1}{2}\um^t M\um & -\um^t M & 1
  \end{BMAT}\right)
\end{split}
\end{equation}
with
\begin{equation}\label{eq:upum-defined}
  \up=
  \left(\begin{array}{cc}
    \wp \ \vp
  \end{array}\right)
    \qquad \mathrm{and} \qquad
  \um=
  \left(\begin{array}{c}
    \wm \\
    \vm
  \end{array}\right) \,,
\end{equation}
cf.~(\ref{eq:vp}, \ref{eq:vm}, \ref{eq:wp}, \ref{eq:wm}), and $\mathcal{L}_{x_1,x_2}(x)$ being the linear
$\fosp(2n-2|2m)$ type Lax matrix of \eqref{eq:linlaxnorm}:
\begin{equation}\label{eq:innerLax2}
  \mathcal{L}_{x_1,x_2}(x) =
  \left(\begin{BMAT}[5pt]{c:c}{c:c}
    (x+x_1)\ID-\kpp_1\kmm_1 & \left((x_2-x_1)\ID+\kpp_1\kmm_1\right)\kpp_1 \\
     -\kmm_1 & (x+x_2)\ID_{}+\kmm_1\kpp_1
  \end{BMAT}\right) \,,
\end{equation}
cf.~\eqref{eq:innerLax}, is a solution to the RTT-relation \eqref{eq:rttA} with the $R$-matrix~\eqref{eq:R-osp}, hence, is a Lax matrix.

\medskip
\noindent
(b) The Lax matrix from (a) is a fusion of two degenerate Lax matrices through:
\begin{equation}
  L^{[1]}(x+x_1)\hat L^{[2]}(x+x_2) =
  \mathbf{S} \, \mathbf{L}_{x_1,x_2}(x)
  \left(\begin{BMAT}[5pt]{c|c:c|c}{c|c:c|c}
    1 & 0 & 0 & 0 \\
    0 & \id & \kp_{2}^\circ & 0 \\
    0 & 0 & \id & 0 \\
    0 & 0 & 0 & 1
  \end{BMAT}\right) \,
  \mathbf{S}^{-1}
\end{equation}
with the Lax matrices $L(x),\hat L(x)$ of (\ref{eq:lindec},~\ref{eq:LaxZwei}) and the
similarity transformation $\mathbf{S}$ of~\eqref{eq:linear+cubic-S}--\eqref{eq:Ksim2}.
\end{Thm}

\begin{Rem}
For $m=0$, we recover precisely the $D_n$-type Lax matrix of~\cite[(5.24, 5.25)]{f} when setting $x_1=s, x_2=-s-n+2$.
\end{Rem}


\subsection{Degenerate quadratic Lax matrices: even $N$ case}
\

Dropping the index 1 in the oscillators of (\ref{eq:vp}, \ref{eq:vm}, \ref{eq:wp}, \ref{eq:wm}),
we consider the vectors
\begin{equation}\label{eq:um}
  \um =
  \left(\begin{array}{ccccccccc}
    \oa_{2} \ \ldots \ \oa_{n} \ \oc_{n+1} \ \cdots \ \oc_{n+m} \ \oc_{n+m+1} \ \cdots \
    \oc_{n+2m} \ \oa_{n+2m+1} \ \cdots \ \oa_{2n+2m-1}
  \end{array}\right)^t
\end{equation}
and
\begin{equation}\label{eq:up}
  \up =
  \left(\begin{array}{ccccccccc}
    \oad_{2} \ \cdots \ \oad_{n} \ \ocd_{n+1} \ \cdots \ \ocd_{n+m} \ \ocd_{n+m+1} \ \cdots \ \ocd_{n+2m} \
    \oad_{n+2m+1} \ \cdots \ \oad_{2n+2m-1}
  \end{array}\right) \,,
\end{equation}
cf.~\eqref{eq:upum-defined}, so that $\oa_j=\oa_{j,1}$, $\oc_j=\oc_{j,1}$, $\oad_j=\oad_{1,j}$, and $\ocd_j=\ocd_{1,j}$.

As an immediate corollary of Theorem~\ref{prop:quadL}, we obtain the following result:

\begin{Prop}\label{prop:quadratic-degenerate-specialized}
The matrix
\begin{equation}\label{eq:quadl22}
{\small
  \mathbf{L}(x)=
  \left(\begin{BMAT}[5pt]{c|c|c}{c|c|c}
    1 & \up & \frac{1}{2}\up M\up^t \\
    0 & \id_{N+2m-2} & M\up^t \\
    0 & 0 & 1
  \end{BMAT}\right)
  \left(\begin{BMAT}[5pt]{c|c|c}{c|c|c}
    x(x-\kappa+1) & 0 & 0 \\
    0 & x\ID_{N+2m-2} & 0 \\
    0 & 0 & 1
  \end{BMAT}\right)
  \left(\begin{BMAT}[5pt]{c|c|c}{c|c|c}
    1 & 0 & 0 \\
    -\um & \id_{N+2m-2} & 0 \\
    \frac{1}{2}\um^t M\um & -\um^tM & 1
  \end{BMAT}\right)
}
\end{equation}
with $\um$ of~\eqref{eq:um}, $\up$ of~\eqref{eq:up}, $M$ of \eqref{eq:Mmatrix}, is a solution to the RTT-relation \eqref{eq:rttA} with
the $R$-matrix~\eqref{eq:R-osp}, hence, is a Lax matrix. It can be explicitly written as
\begin{equation}\label{eq:quadlaxl41}
  \mathbf{L}(x)=
  \left(\begin{BMAT}[5pt]{c|c|c}{c|c|c}
    x(x-\kappa+1)-x\up\um+\frac{1}{4}\up M\up^t\um^t M\um & x\up-\frac{1}{2}\up M\up^t\um^t M & \frac{1}{2}\up M\up^t \\
    -x\um+\frac{1}{2}M\up^t\um^t M\um & x\ID_{N+2m-2}-M\up^t\um^t M & M\up^t \\
    \frac{1}{2}\um^t M\um & -\um^t M & 1
  \end{BMAT}\right) \,.
\end{equation}
\end{Prop}

\begin{proof}
According to Theorem~\ref{prop:quadL}, the matrix $\mathbf{L}_{x_1,x_2}(x)$ of \eqref{eq:quadl21} is a solution to
the RTT-relation \eqref{eq:rttA} with the $R$-matrix~\eqref{eq:R-osp}. The superoscillators used in $\um$ and $\up$
pairwise supercommute with the ones used in the construction of $\mathcal{L}_{x_1,x_2}(x)$.
Furthermore, the coefficients of the latter generate an orthosymplectic subalgebra. To prove that the matrix $\mathbf{L}(x)$
in \eqref{eq:quadl22} is a solution to the same RTT-relation, we consider the trivial representation of the aforementioned
orthosymplectic subalgebra. It is obtained from the action of the Lax matrix $\mathcal{L}_{x_1,x_2}(x)$ on the Fock vacuum
$|0\rangle$ (in the Fock module for the superoscillator algebra generated by the entries of $\kmm,\kpp$ from~\eqref{eq:Koscillators2},
so that $\kmm|0\rangle$=0) and further fixing $x_2=x_1-\kappa+1$ with $\kappa=n-m-1$ as in~\eqref{eq:kappa}. Using the simple equalities
\begin{equation}
  \kmm\kpp |0\rangle = (\kappa-1)\, \id |0\rangle \,,\qquad
  \kpp\kmm\kpp |0\rangle = (\kappa-1)\, \kpp|0\rangle  \,,
\end{equation}
we find
\begin{equation}
  \mathcal{L}_{x_1,x_2=x_1-\kappa+1}(x )|0\rangle = (x+x_1)\id_{N+2m-2}|0\rangle \,.
\end{equation}
Specializing further $x_1=0$, we conclude that \eqref{eq:quadlaxl41} is indeed an orthosymplectic Lax matrix.
\end{proof}

\begin{Rem}
For $m=0$, we recover precisely the $D_n$-type Lax matrix of~\cite[(4.12)]{f}.
\end{Rem}


\subsection{Degenerate quadratic Lax matrices: odd $N$ case}\label{ssec:odd case}
\

In this subsection, we consider the case $N=2n+1$. In this setup, the operators $\Pop,\Qop$ are defined as
in~(\ref{eq:Pop},~\ref{eq:Qop}), where we choose the following specific $\BZ_2$-grading of the superspace $V$:
\begin{equation}\label{eq:our-osp-grading-odd}
  \gr{i}:=|v_i|=
  \begin{cases}
    \bar{0} & \quad \text{for} \ 1\leq i\leq n \\
    \bar{1} & \quad \text{for} \ n+1\leq i\leq n+m \\
    \bar{0} & \quad \text{for} \ i=n+m+1 \\
    \bar{1} & \quad \text{for} \ n+m+2\leq i\leq n+2m+1 \\
    \bar{0} & \quad \text{for} \ n+2m+2\leq i\leq 2n+2m+1
   \end{cases} \,.
\end{equation}
This $\BZ_2$-grading corresponds to the parity sequence
\begin{equation}\label{eq:our-parity-odd}
  \Parity =
  \Big(
    \underbrace{\bar{0},\ldots,\bar{0}}_{n}, \underbrace{\bar{1},\ldots,\bar{1}}_{m}, \bar{0},
    \underbrace{\bar{1},\ldots,\bar{1}}_{m}, \underbrace{\bar{0},\ldots,\bar{0}}_{n}
  \Big)
\end{equation}
and the following choice of $\theta_i$'s:
\begin{equation}\label{eq:theta-odd}
  \theta=\theta_V=\Big(\underbrace{1,\ldots,1}_{n+m+1},\underbrace{-1,\ldots,-1}_{m},\underbrace{1,\ldots,1}_{n}\Big) \,.
\end{equation}

In this case, we upgrade~(\ref{eq:um},~\ref{eq:up}) by adding extra bosonic superoscillators:
\begin{equation}\label{eq:um-odd}
  \um=
  \left(\begin{array}{ccccccccc}
    \oa_{2} \ \ldots \ \oa_{n} \ \oc_{n+1} \ \cdots \ \oc_{n+m} \ \oa_{n+m+1} \
    \oc_{n+m+2} \ \ldots \ \oc_{n+2m+1} \ \oa_{n+2m+2} \ \cdots \ \oa_{2n+2m}
  \end{array}\right)^t \,,
\end{equation}
\begin{equation}\label{eq:up-odd}
  \up =
  \left(\begin{array}{ccccccccc}
    \oad_{2} \ \cdots \ \oad_{n} \ \ocd_{n+1} \ \cdots \ \ocd_{n+m} \ \oad_{n+m+1} \
    \ocd_{n+m+2} \ \cdots \ \ocd_{n+2m+1} \ \oad_{n+2m+2} \ \cdots \ \oad_{2n+2m}
  \end{array}\right) \,.
\end{equation}
We also modify~\eqref{eq:Mmatrix} by introducing the following $(N+2m-2)\times (N+2m-2)$ matrix:
\begin{equation}\label{eq:Mmatrix-odd}
  M=
  \left(\begin{BMAT}[5pt]{c:c:c}{c:c:c}
    0 & 0 & -\idb \\
    0 & -1 & 0 \\
    \idg^{-1} & 0 & 0
  \end{BMAT}\right) \,.
\end{equation}

\begin{Conj}\label{conj:odd-N}
The matrix $\mathbf{L}(x)$ given by~(\ref{eq:quadl22},~\ref{eq:quadlaxl41}) with $\um$ of~\eqref{eq:um-odd},
$\up$ of~\eqref{eq:up-odd}, $M$ of~\eqref{eq:Mmatrix-odd} is a solution to the RTT-relation~\eqref{eq:rttA} with
the $R$-matrix~\eqref{eq:R-osp}, hence, is a Lax matrix.
\end{Conj}

\begin{Rem}
(a) This has been confirmed for $n,m\leq 2$, but a general proof is currently missing.

\medskip
\noindent
(b) For $m=0$, we recover precisely the $B_n$-type Lax matrix of~\cite[(9.1)]{fkt}.
\end{Rem}


\subsection{Non-degenerate quadratic Lax matrix through the factorisation}\label{ss:facq}
\

In this subsection, we present a factorisation formula for quadratic Lax matrices that yields a non-degenerate
quadratic orthosymplectic Lax matrix of superoscillator type. We uniformly treat both cases of even
and odd $N$, assuming the validity of Conjecture~\ref{conj:odd-N}.

The factorisation is analogous to that for the linear case from Subsection~\ref{sec:Facllm}. As before,
we first introduce another solution to the same RTT-relation via a proper conjugation of the degenerate Lax matrix
$\mathbf{L}(x)$ from Proposition~\ref{prop:quadratic-degenerate-specialized} or Conjecture~\ref{conj:odd-N}.
To this end, recall the matrix $\tilde \idb$ of~\eqref{eq:tilde-idb-matrix}:
\begin{equation}
  \tilde \idb=
  \left(\begin{BMAT}[5pt]{c|c|c}{c|c|c}
    0 & 0 & 1 \\
    0 & \id_{N+2m-2} & 0 \\
    1 & 0 & 0
  \end{BMAT}\right) \,.
\end{equation}
The $R$-matrix~\eqref{eq:R-osp} commutes with $\tilde \idb \otimes \tilde \idb$:
for even $N$ this is proved in Lemma~\ref{lem:tj}, while for odd $N$ the argument is the same.
Thus, we get another solution to the same RTT-relation \eqref{eq:rttA} via
\begin{equation}
  \tilde {\mathbf{L}}(x) = \tilde \idb \mathbf{L}(x) \tilde \idb^{-1} = \tilde \idb \mathbf{L}(x) \tilde \idb \,.
\end{equation}
We further apply the particle-hole transformation
\begin{equation}
  \oa_i \mapsto \oad_{i'} \,,\qquad  \oad_i\mapsto -\oa_{i'} \,,
\end{equation}
\begin{equation}
  \oc_i\mapsto \theta_i\ocd_{i'} \,,\qquad  \ocd_i\mapsto \theta_i\oc_{i'} \,,
\end{equation}
to obtain the following Lax matrix:
\begin{multline}\label{eq:quadlaxl5}
  \hat{\tilde{\mathbf{L}}}(y) =
  \tilde{\mathbf{L}}(y)|_{p.h.} = \\
  \left(\begin{BMAT}[5pt]{c|c|c}{c|c|c}
    1 & \up & \frac{1}{2}\up M\up^t \\
    \um & y\id_{N+2m-2}+\um\up & y M\up^t+\frac{1}{2}\um\up M\up^t \\
    \frac{1}{2}\um^tM\um & y\um^tM+\frac{1}{2}\um^tM\um\up & y(y-\kappa+1)+y\um^tMM\up^t+\frac{1}{4}\um^tM\um\up M\up^t
  \end{BMAT} \right) \,.
\end{multline}

Similarly to~(\ref{eq:linear factorization1}),~(\ref{eq:linear factorization}), and (\ref{eq:factorisatio-no3}),
let us factorise the product of the Lax matrices ${\mathbf{L}}^{[1]}(x)$ and~$\widehat{\tilde{\mathbf{L}}}^{[2]}(y)$
from \eqref{eq:quadlaxl41} and~\eqref{eq:quadlaxl5}, with two families of superoscillators encoded by $(\um_1,\up_1)$
and $(\um_2,\up_2)$, respectively. Explicitly, one obtains the following factorisation:
\begin{equation}
  {\mathbf{L}}^{[1]}(x)\hat{\tilde{\mathbf{L}}}^{[2]}(y) = {\mathfrak{L}}'(x,y) H'
\end{equation}
with
\begin{equation}\label{Matrix Example D4 a-osc factroized}
  \mathfrak{L}'(x,y)=
    \left(\begin{BMAT}[5pt]{c|c|c}{c|c|c}
      1 & \up'_1 & \tfrac{1}{2}\up'_1 M (\up'_1)^t \\
      0 & \ID_{N+2m-2} &M(\up'_1)^t \\
      0 & 0 & 1
    \end{BMAT}\right)
      \cdot \, D'(x,y) \, \cdot
    \left(\begin{BMAT}[5pt]{c|c|c}{c|c|c}
      1 & -\up'_1 & \tfrac{1}{2}\up'_1 M(\up'_1)^t \\
      0 & \ID_{N+2m-2} & -M(\up'_1)^t \\
      0 & 0 & 1
    \end{BMAT}\right)
\end{equation}
where
\begin{equation*}
  D'(x,y)=
    \left(\begin{BMAT}[5pt]{c|c|c}{c|c|c}
    x(x-\kappa+1) & 0 & 0 \\
      -x\um_1' & x y\ID_{N+2m-2} & 0 \\
      \tfrac{1}{2}(\um'_1)^t M\um_1' &-y (\um'_1)^t M & y(y-\kappa+1)
    \end{BMAT}\right) \,,
\end{equation*}
$M$ is given by~\eqref{eq:Mmatrix} for $N=2n$ or by~\eqref{eq:Mmatrix-odd} for $N=2n+1$, and
\begin{equation}
  H'=
  \left(\begin{BMAT}[5pt]{c:c:c}{c:c:c}
    1 & \up_2' & \frac{1}{2}\up_2'M(\up'_2)^t \\
    0 & \id_{N+2m-2} & M (\up'_2)^t \\
    0 & 0 & 1
  \end{BMAT}\right) \,.
\end{equation}
Here, we introduced the following notation:
\begin{equation}
   \um_1'=\um_1-\um_2 \,,\qquad \up_1'=\up_1 \,,
\end{equation}
and
\begin{equation}
    \up_2'=\up_2+\up_1 \,,\qquad \um_2'=\um_2 \,,
\end{equation}
cf.~\eqref{eq:tra121}. Furthermore, we used the following simple properties:
\begin{equation}
 \up_2 M\up_1^t=\up_1 M \up_2^t \,,\qquad \um_1^tM\um_2=\um_2^t M \um_1 \,.
\end{equation}

Finally, we note that there is a similarity transformation in the superoscillator space such that
\begin{equation}
  \up_i' = \mathbf{S}\up_i  \mathbf{S}^{-1} \,, \qquad
  \um_i' = \mathbf{S}\um_i  \mathbf{S}^{-1} \qquad (i=1,2) \,,
\end{equation}
cf.~(\ref{eq:Atype-similarity},~\ref{eq:S-Atype}). For even $N=2n$, it reads
\begin{equation}\label{eq:Kma3-eve}
  \mathbf{S}=
  \exp\left[ \sum_{i=2}^n\oad_i^{[1]}\oa_i^{[2]} + \sum_{i=n+1}^{n+m}\ocd_i^{[1]}\oc_i^{[2]} +
  \sum_{i=n+m+1}^{n+2m}\ocd_i^{[1]}\oc_i^{[2]}+\sum_{i=n+2m+1}^{2n+2m-1}\oad_i^{[1]}\oa_i^{[2]}\right] \,,
\end{equation}
while for odd $N=2n+1$, we have
\begin{equation}\label{eq:Kma3}
  \mathbf{S}=
  \exp\left[ \sum_{i=2}^n\oad_i^{[1]}\oa_i^{[2]} + \sum_{i=n+1}^{n+m}\ocd_i^{[1]}\oc_i^{[2]} +
  \oad_{n+m+1}^{[1]}\oa_{n+m+1}^{[2]}+\sum_{i=n+m+2}^{n+2m+1}\ocd_i^{[1]}\oc_i^{[2]}+
  \sum_{i=n+2m+2}^{2n+2m}\oad_i^{[1]}\oa_i^{[2]}\right] \,.
\end{equation}

We thus obtain the key result of this subsection:

\begin{Prop}\label{prop:nondeg-quadratic-small}
For any $x_1,x_2\in \BC$, the matrix
\begin{equation}\label{Matrix Example D4 a-osc factroized5}
  \mathcal{L}_{x_1,x_2}(x)=
  \left(\begin{BMAT}[5pt]{c|c|c}{c|c|c}
    1 & \up_1 & \tfrac{1}{2}\up_1 M \up_1^t \\
    0 & \ID_{N+2m-2} &M\up_1^t \\
    0 & 0 & 1 \\
  \end{BMAT}\right)
    \cdot \, D_{x_1,x_2}(x) \, \cdot
  \left(\begin{BMAT}[5pt]{c|c|c}{c|c|c}
    1 & -\up_1 & \tfrac{1}{2}\up_1 M\up_1^t \\
    0 & \ID_{N+2m-2} & -M\up_1^t \\
    0 & 0 & 1 \\
  \end{BMAT}\right)
\end{equation}
with
\begin{equation*}
\small{
  D_{x_1,x_2}(x)=
  \left(\begin{BMAT}[5pt]{c|c|c}{c|c|c}
    (x+ x_1)(x+x_1-\kappa+1) & 0 & 0 \\
    -(x+x_1)\um_1 & (x+x_1)(x+x_2)\ID_{N+2m-2} & 0 \\
    \tfrac{1}{2}\um_1^tM\um_1 & -(x+x_2) \um_1^tM & (x+x_2)(x+x_2-\kappa+1)
  \end{BMAT}\right)
}
\end{equation*}
is a solution to the RTT-relation \eqref{eq:rttA} with the $R$-matrix~\eqref{eq:R-osp}, hence, is a Lax matrix.
Furthermore, it obeys the factorisation formula
\begin{equation}
  {\mathbf{L}}^{[1]}(x+x_1) \hat{\tilde{\mathbf{L}}}^{[2]}(x+x_2) =
  \mathbf{S} \mathcal{L}_{x_1,x_2}(x)H \mathbf{S}^{-1}
\end{equation}
with ${\mathbf{L}}(x)$ of \eqref{eq:quadlaxl41}, $\hat{\tilde{\mathbf{L}}}(y)$ of \eqref{eq:quadlaxl5},
$\mathbf{S}$ of (\ref{eq:Kma3-eve},~\ref{eq:Kma3}), and $H$ given by
\begin{equation}\label{eq:H-end}
  H=
  \left(\begin{BMAT}[5pt]{c:c:c}{c:c:c}
    1 & \up_2 & \frac{1}{2}\up_2M\up_2^t \\
    0 & \id_{N+2m-2} & M \up_2^t \\
    0 & 0 & 1
  \end{BMAT}\right) \,.
\end{equation}
\end{Prop}

\begin{Rem}
For $m=0$, we recover the $\mathfrak{so}_N$-type Lax matrix $\mathfrak{L}_{-x_1,-x_2}(x)$ of \cite[(7.3)]{fkt}.
\end{Rem}


\subsection{Full non-degenerate quadratic Lax matrix through the factorisation}
\

For even $N=2n$, the constructions from the previous subsection admit straightforward generalizations
when replacing $\mathbf{L}(x)$ with $\mathbf{L}_{x_1,x_2}(x)$ of \eqref{eq:quadl21}:

\begin{Thm}\label{prop:nondeg-quadratic-big}
For $N=2n$ and any $x_1,x_2,y_1,y_2\in\mathbb{C}$, the matrix
\begin{equation}\label{Matrix Example D4 a-osc factroizede5}
  \mathcal{L}_{x_1,x_2,y_1,y_2}(x)=
  \left(\begin{BMAT}[5pt]{c|c|c}{c|c|c}
    1 & \up_1 & \tfrac{1}{2}\up_1 M \up_1^t \\
    0 & \ID_{2n+2m-2} & M\up_1^t \\
    0 & 0 & 1
  \end{BMAT}\right)
    \cdot \, D_{x_1,x_2,y_1,y_2}(x) \, \cdot
  \left(\begin{BMAT}[5pt]{c|c|c}{c|c|c}
    1 & -\up_1 & \tfrac{1}{2}\up_1 M\up_1^t \\
    0 & \ID_{2n+2m-2} & -M\up_1^t \\
    0 & 0 & 1
  \end{BMAT}\right)
\end{equation}
with
\begin{equation*}
\small{
  D_{x_1,x_2,y_1,y_2}(x)=
  \left(\begin{BMAT}[5pt]{c|c|c}{c|c|c}
    (x+y_1)(x+y_1-\kappa+1) & 0 & 0 \\
    -\mathcal{L}_{x_1,x_2}(x+y_1)\um_{1} & (x+y_2)\mathcal{L}_{x_1,x_2}(x+y_1) & 0 \\
    \tfrac{1}{2}\um_1^tM\um_1 & -(x+y_2) \um_1^tM & (x+y_2)(x+y_2-\kappa+1) \\
  \end{BMAT}\right)
}
\end{equation*}
and $\mathcal{L}_{x_1,x_2}(x)$ being the linear $\fosp(2n-2|2m)$ type Lax matrix of \eqref{eq:innerLax2},
is a solution to the RTT-relation \eqref{eq:rttA} with the $R$-matrix~\eqref{eq:R-osp}, hence, is a Lax matrix.
Furthermore, it obeys the factorisation formula
\begin{equation}
  \mathbf{L}_{x_1,x_2}^{[1]}(x+y_1) \hat{\tilde{\mathbf{L}}}^{[2]}(x+y_2) =
  \mathbf{S} \mathcal{L}_{x_1,x_2,y_1,y_2}(x)H \mathbf{S}^{-1}
\end{equation}
with $\hat{\tilde{\mathbf{L}}}(x)$ of \eqref{eq:quadlaxl5}, $H$ of~\eqref{eq:H-end},
and $\mathbf{L}_{x_1,x_2}(x)$ of \eqref{eq:quadl21}.
\end{Thm}

\begin{Rem}
For $m=0$, we recover the $D_n$-type Lax matrix $\mathfrak{L}_{n,s}(x)$ of~\cite[(5.36, 5.37)]{f}, depending on the extra
parameters ${\textrm x}_1,{\textrm x}_2$, when setting $x_1=s, x_2=-s-n+2$, $y_1=-{\textrm x}_1, y_2=-{\textrm x}_2$.
\end{Rem}

\begin{Rem}\label{rem:quadratic-Lax-all-parities}
According to Remark~\ref{rem:different parities osp}, the Lax matrix $\mathbf{L}_{x_1,x_2}(x)$ from
Theorem~\ref{prop:quadL}(a), the Lax matrix $\mathbf{L}(x)$ from Proposition~\ref{prop:quadratic-degenerate-specialized},
and its generalization for odd $N$ from Conjecture~\ref{conj:odd-N} give rise to the corresponding degenerate quadratic
Lax matrices for any other $\BZ_2$-grading satisfying $|v_1|=\bar{0}$. Likewise, the Lax matrices
$\mathcal{L}_{x_1,x_2}(x)$ from Proposition~\ref{prop:nondeg-quadratic-small} and $\mathcal{L}_{x_1,x_2,y_1,y_2}(x)$
from Theorem~\ref{prop:nondeg-quadratic-big} give rise to the corresponding non-degenerate quadratic Lax matrices for
all $\BZ_2$-gradings of $V$.
\end{Rem}

\begin{Rem}\label{rem:normalized-limit-osp-quadratic}
Similarly to Remarks~\ref{rem:Atype-normalized-limits},~\ref{rem:normalized-limit-osp}, we can obtain
the degenerate Lax matrices \eqref{eq:quadl22} and \eqref{eq:quadlaxl5} from the non-degenerate Lax matrix
\eqref{Matrix Example D4 a-osc factroized5} through the \emph{renormalized limit} procedure:
\begin{equation}\label{eq:limit1}
\begin{split}
  & \mathbf{L}(x+x_1) =
    \lim_{x_2\to \infty}\ \Big\{\mathcal{L}_{x_1,x_2}(x)\cdot \mathrm{diag}
    \Big(1,\underbrace{x_2^{-1},\ldots,x_2^{-1}}_{N+2m-2},x_2^{-2}\Big)\Big\}\Big|_{\um_1 \mapsto \um \,,\, \up_1\mapsto \up} \,, \\
  & \hat{\tilde{\mathbf{L}}}(x+x_2) = \lim_{x_1\to \infty}\ \Big\{\mathrm{diag}
    \Big(x_1^{-2},\underbrace{x_1^{-1},\ldots ,x_1^{-1}}_{N+2m-2},1\Big)
    \cdot\mathcal{L}_{x_1,x_2}(x)\Big\}\Big|_{\um_1 \mapsto -\um \,,\, \up_1\mapsto -\up} \,.
\end{split}
\end{equation}
Analogously, for even $N=2n$, the degenerate Lax matrices \eqref{eq:quadl21} and \eqref{eq:quadlaxl5} can be obtained as the
\emph{renormalized limits} of the full non-degenerate quadratic Lax matrix $\mathcal{L}_{x_1,x_2,y_1,y_2}(x)$ from \eqref{Matrix Example D4 a-osc factroizede5}.
\end{Rem}

%

%

\appendix

\section{Twists for transfer matrices and Q-operators}\label{sec:twists}
In analogy to $BCD$-type spin chain transfer matrices, see~\cite[(8.69, 9.19)]{fkt}, the diagonal
twists $D$ for the finite-dimensional transfer matrices should be of the following form:
\begin{itemize}

\item
For $N=2n$
\begin{equation}\label{eq:T-twist-even}
  D=\diag \Big(\tau_1,\ldots,\tau_n,\tau_{n+1},\ldots,\tau_{n+m},\tau_{n+m}^{-1},\ldots,\tau_{n+1}^{-1},\tau_{n}^{-1},\ldots,\tau_{1}^{-1}\Big)
\end{equation}

\item
For $N=2n+1$
\begin{equation}\label{eq:T-twist-odd}
  D=\diag \Big(\tau_1,\ldots,\tau_n,\tau_{n+1},\ldots,\tau_{n+m},1,\tau_{n+m}^{-1},\ldots,\tau_{n+1}^{-1},\tau_{n}^{-1},\ldots,\tau_{1}^{-1}\Big)
\end{equation}

\end{itemize}

On the other hand, the twist $D_{\ossc}$ used to construct $Q$-operators from the monodromy matrices
\begin{equation*}
  M(x)=\underbrace{L(x)\otimes \cdots \otimes L(x)}_{N\ \mathrm{times}}
\end{equation*}
of degenerate Lax matrices $L(x)$ is determined through the following \emph{invariance condition}:
\begin{equation}
  D L(x) D^{-1} = D^{-1}_{\ossc} L(x) D_{\ossc} \,,
\end{equation}
cf.~\cite[(8.68, 9.18)]{fkt}. Here, the twist $D$ of~(\ref{eq:T-twist-even},~\ref{eq:T-twist-odd})
acts only on the matrix space, while the twist $D_{\ossc}$ acts only on the oscillator space.
We thus derive the following explicit formulas:
\begin{itemize}

\item
For $N=2n$ and the linear Lax matrix of \eqref{eq:linlaxdeg}, $D_{\ossc}$ is given by
\begin{equation}\label{eq:Q-twist-1}
  D_{\ossc}=
  \left(\prod_{1\leq i< j\leq n }(\tau_i\tau_j)^{-\oad_{ij'}\oa_{j'i}}\right)
  \left(\prod_{n+1\leq i\leq j \leq n+m }(\tau_i\tau_j)^{-\obd_{ij'}\ob_{j'i}}\right)
  \left(\prod_{i=1}^n\prod_{j=n+1}^{n+m}(\tau_i\tau_j)^{-\ocd_{ij'}\oc_{j'i}}\right)
\end{equation}
which is a mixture of the corresponding $D$-type and $C$-type formulas~\cite[(8.66, 8.90)]{fkt}

\medskip
\item
For $N=2n$ and the quadratic Lax matrix of~(\ref{eq:quadl22},~\ref{eq:quadlaxl41}), $D_{\ossc}$ is given by
\begin{equation}\label{eq:Q-twist-2}
  D_{\ossc}=
  \tau_1^{-\sum_{j=2}^{n}(\oad_j\oa_j+\oad_{j'}\oa_{j'})-\sum_{j=n+1}^{n+m}(\ocd_j\oc_j+\ocd_{j'}\oc_{j'})}
  \left(\prod_{j=2}^{n}\tau_j^{\oad_j\oa_j-\oad_{j'}\oa_{j'}}\right)\left(\prod_{j=n+1}^{n+m}\tau_j^{\ocd_j\oc_j-\ocd_{j'}\oc_{j'}}\right)
\end{equation}
which is an analogue of~\cite[(9.16)]{fkt}

\medskip
\item
For $N=2n+1$ and the quadratic Lax matrix of Conjecture~\ref{conj:odd-N}, $D_{\ossc}$ is alike given by
\begin{equation}\label{eq:Q-twist-3}
\begin{split}
  & D_{\ossc}=\\
  & \tau_1^{-\oad_{n+m+1}\oa_{n+m+1}-\sum_{j=2}^{n}(\oad_j\oa_j+\oad_{j'}\oa_{j'})-\sum_{j=n+1}^{n+m}(\ocd_j\oc_j+\ocd_{j'}\oc_{j'})}
  \left(\prod_{j=2}^{n}\tau_j^{\oad_j\oa_j-\oad_{j'}\oa_{j'}}\right)\left(\prod_{j=n+1}^{n+m}\tau_j^{\ocd_j\oc_j-\ocd_{j'}\oc_{j'}}\right)
\end{split}
\end{equation}

\end{itemize}



\begin{thebibliography}{99}

\bibitem[AACFR]{Arnaudon}
D.~Arnaudon, J.~Avan, N.~Cramp\'{e}, L.~Frappat, E.~Ragoucy,
  {\em $R$-matrix presentation for super-Yangians $Y(osp(m|2n))$},
J.\ Math.\ Phys.\ {\bf 44} (2003), no.~1, 302--308.

\bibitem[BCFGT]{bcf}
D.~Bombardelli, A.~Cavaglià, D.~Fioravanti, N.~Gromov, R.~Tateo,
  {\em The full quantum spectral curve for $AdS_4$/$CFT_3$},
J.\ High Energ.\ Phys.\ {\bf 140} (2017), no.~9, 72pp.

\bibitem[BFN]{BFN}
A.~Braverman, M.~Finkelberg, H.~Nakajima,
  {\em Coulomb branches of $3d$ $\mathcal{N}=4$ quiver gauge theories and slices in the affine Grassmannian}
(with appendices by A.~Braverman, M.~Finkelberg, J.~Kamnitzer, R.~Kodera, H.~Nakajima, B.~Webster, A.~Weekes),
Adv.\ Theor.\ Math.\ Phys.\ {\bf 23} (2019), no.~1, 75--166.

\bibitem[BHK]{BHK}
V.~Bazhanov, A.~Hibberd, S.~Khoroshkin,
  {\em Integrable structure of $W_3$ conformal field theory, quantum Boussinesq theory and boundary affine Toda theory},
Nuclear Phys.\ B {\bf 622} (2002), no.~3, 475--547.

\bibitem[BK]{bk}
J.~Brundan, A.~Kleshchev,
  {\em Parabolic presentations of the Yangian $Y(\gl_n)$},
Commun.\ Math.\ Phys.\ {\bf 254} (2005), no.~1, 191--220.

\bibitem[BLZ]{BLZ}
V.~Bazhanov, S.~Lukyanov, A.~Zamolodchikov,
  {\em Integrable structure of conformal field theory II. Q-operator and DDV equation},
Commun.\ Math.\ Phys.\ {\bf 190} (1997), no.~2, 247--278.

\bibitem[BT]{Bazhanov:2008yc}
V.~Bazhanov, Z.~Tsuboi,
  {\em Baxter's $Q$-operators for supersymmetric spin chains},
Nuclear Phys.\ B {\bf 805} (2008), no.~3, 451--516.

\bibitem[CGY]{cgy}
K.~Costello, D.~Gaiotto, J.~Yagi,
  {\em $Q$-operators are  $'\mathrm{t}$ Hooft lines},
preprint, ar$\chi$iv:2103.01835 (2021).

\bibitem[Fr1]{f}
R.~Frassek,
  {\em Oscillator realisations associated to the $D$-type Yangian: towards the operatorial $Q$-system of orthogonal spin chains},
Nuclear Phys.\ B {\bf 956} (2020), Paper No.~115063, 22pp.

\bibitem[Fr2]{ftalk}
R.~Frassek,
  {\em \href{https://agenda.infn.it/event/33911/contributions/207941/}{Lax matrices for Baxter Q-operators}},
Seminar at the University of Bologna (Sept.~4, 2023), 10th Bologna Workshop on Conformal Field Theory
and Integrable Models (Sept.~4-7, 2023).

\bibitem[FIKK]{fikk}
J.~Fuksa, A.~Isaev, D.~Karakhanyan, R.~Kirschner,
  {\em Yangians and Yang–Baxter $R$-operators for ortho-symplectic superalgebras},
Nuclear Phys.\ B {\bf 917} (2017), 44--85.

\bibitem[FKT]{fkt}
R.~Frassek, I.~Karpov, A.~Tsymbaliuk,
  {\em Transfer matrices of rational spin chains via novel BGG-type resolutions},
Commun.\ Math.\ Phys.\ {\bf 400} (2023), no.~1, 1--82.

\bibitem[FLMS]{Frasseksusy}
R.~Frassek, T.~Lukowski, C.~Meneghelli, M.~Staudacher,
  {\em Oscillator construction of $\mathfrak{su}(n|m)$ Q-operators},
Nuclear Phys.\ B {\bf 850} (2011), no.~1, 175--198.

\bibitem[FP]{fp}
R.~Frassek, V.~Pestun,
  {\em A family of ${\mathrm GL}_r$ multiplicative Higgs bundles on rational base},
SIGMA {\bf 15} (2019), Paper No.~031, 42pp.

\bibitem[FPT]{fpt}
R.~Frassek, V.~Pestun, A.~Tsymbaliuk,
  {\em Lax matrices from antidominantly shifted Yangians and quantum affine algebras: A-type},
Adv.\ Math.\ {\bf 401} (2022), Paper No.~108283, 73pp.

\bibitem[FT1]{ft1}
R.~Frassek, A.~Tsymbaliuk,
  {\em Rational Lax matrices from antidominantly shifted extended Yangians: BCD types},
Commun.\ Math.\ Phys.\ {\bf 392} (2022), no.~2, 545--619.

\bibitem[FT2]{ft2}
R.~Frassek, A.~Tsymbaliuk,
  {\em Orthosymplectic Yangians},
preprint, ar$\chi$iv:2311.18818 (2023).

\bibitem[G]{gowthesis}
L.~Gow,
  {\em Yangians of Lie superalgebras},
PhD Thesis (2007), University of Sydney.

\bibitem[GM]{GM}
W.~Galleas, M.~Martins,
  {\em $R$-matrices and spectrum of vertex models based on superalgebras},
Nuclear Phys.\ B {\bf 699} (2004), no.~3, 455--486.

\bibitem[HZ]{hz}
D.~Hernandez, H.~Zhang,
  {\em Shifted Yangians and polynomial R-matrices},
preprint, ar$\chi$iv:2103.10993 (2021).

\bibitem[IKK]{ikk}
A.~Isaev, D.~Karakhanyan, R.~Kirschner,
  {\em Yang-Baxter $R$-operators for osp superalgebras},
Nuclear Phys.\ B {\bf 965} (2021), Paper No.~115355, 28pp.

\bibitem[K]{K}
P.~Kulish,
  {\em Integrable graded magnets},
Zap.\ Nauchn.\ Sem.\ Leningrad.\ Otdel.\ Mat.\ Inst.\ Steklov.\ (LOMI) {\bf 145} (1985), 140--163.

\bibitem[KLT]{klt}
V.~Kazakov, S.~Leurent, Z.~Tsuboi,
  {\em Baxter’s $Q$-operators and operatorial B\"{a}cklund flow for quantum (super)-spin chains},
Commun.\ Math.\ Phys.\ {\bf 311} (2012), no.~3, 787--814.

\bibitem[KNS]{kuniba}
A.~Kuniba, T.~Nakanishi, J.~Suzuki,
  {\em $T$-systems and $Y$-systems in integrable systems},
J.\ Phys.\ A {\bf 44} (2011), no.~10, Paper No.~103001, 146pp.

\bibitem[KS]{KS}
P.~Kulish, E.~Sklyanin,
  {\em Solutions of the Yang-Baxter equation},
Zap.\ Nauchn.\ Sem.\ Leningrad.\ Otdel.\ Mat.\ Inst.\ Steklov.\ (LOMI) {\bf 95} (1980), 129--160.

\bibitem[KSZ]{KSZ}
V.~Kazakov, A.~Sorin, A.~Zabrodin,
  {\em Supersymmetric Bethe ansatz and Baxter equations from discrete Hirota dynamics},
Nuclear Phys.\ B {\bf 790} (2008), no.~3, 345--413.

\bibitem[M]{m}
A.~Molev,
  {\em A Drinfeld-type presentation of the orthosymplectic Yangians},
Alg.\ Represent.\ Theory (2023), 26pp.

\bibitem[MR]{mr}
A.~Molev, E.~Ragoucy,
  {\em Gaussian generators for the Yangian associated with the Lie superalgebra $\fosp(1|2m)$},
J.~Algebra (2024), doi:10.1016/j.jalgebra.2024.01.010, 36pp.

\bibitem[MV]{MV}
C.~Marboe, D.~Volin,
  {\em Fast analytic solver of rational Bethe equations},
J.\ Phys.\ A {\bf 50} (2017), no.~20, Paper No.~204002, 14pp.

\bibitem[N]{N}
M.~Nazarov,
  {\em Quantum Berezinian and the classical Capelli identity},
Lett.\ Math.\ Phys.\ {\bf 21} (1991), no.~2, 123--131.

\bibitem[Ts1]{tsuboi1}
Z.~Tsuboi,
  {\em Analytic Bethe ansatz and functional equations for Lie superalgebra $\ssl(r+1|s+1)$},
J.\ Phys.\ A {\bf 30} (1997), no.~22, 7975--7991.

\bibitem[Ts2]{tsuboi2}
Z.~Tsuboi,
  {\em Analytic Bethe ansatz and functional equations associated with any simple root systems of the Lie superalgebra $\ssl(r+1|s+1)$},
Phys.\  A {\bf 252} (1998), no.~3-4, 565--585.

\bibitem[Ts3]{tsuboi3}
Z.~Tsuboi,
  {\em Solutions of the $T$-system and Baxter equations for supersymmetric spin chains},
Nuclear Phys.\ B {\bf 826} (2010), no.~3, 399--455.

\bibitem[Ts4]{Tsuboi:2012sz}
Z.~Tsuboi,
  {\em Asymptotic representations and $q$-oscillator solutions of the graded Yang-Baxter equation related to Baxter $Q$-operators},
Nuclear Phys.\ B {\bf 886} (2014), 1--30.

\bibitem[Ts5]{Tsuboi:2019vvv}
Z.~Tsuboi,
  {\em A note on $q$-oscillator realizations of $U_{q}(gl(M|N))$ for Baxter $Q$-operators},
Nuclear Phys.\ B {\bf 947} (2019), Paper No.~114747, 33pp.

\bibitem[Ts6]{Tsuboi:2023}
Z.~Tsuboi,
  {\em Folding QQ-relations and transfer matrix eigenvalues: towards a unified approach to Bethe ansatz for super spin chains},
preprint, ar$\chi$iv:2309.16660 (2023).

\bibitem[ZZ]{zz}
A.~Zamolodchikov, A.~Zamolodchikov,
  {\em Factorized $S$-matrices in two dimensions as the exact solutions of certain relativistic quantum field theory models},
Ann.\ Physics {\bf 120} (1979), no.~2, 253--291.

\end{thebibliography}
\end{document}